\documentclass[11pt]{article}

\usepackage{amsmath,amssymb,amsthm}
\usepackage{geometry}
\usepackage{hyperref}
\usepackage{booktabs}
\usepackage{graphicx}
\usepackage{tikz}
\usepackage{quantikz}
\usepackage{placeins}
\graphicspath{{figures/}}

\newcommand{\N}{\mathbb{N}}
\newcommand{\R}{\mathbb{R}}
\newcommand{\C}{\mathbb{C}}

\newcommand{\cO}{\mathcal{O}}
\newcommand{\cK}{\mathcal{K}}
\newcommand{\cH}{\mathcal{H}}
\newcommand{\cX}{\mathcal{X}}
\newcommand{\cA}{\mathcal{A}}

\newcommand{\lift}{\mathrm{lift}}
\newcommand{\spire}{\mathrm{spire}}
\newcommand{\tower}{\mathrm{tower}}
\newcommand{\Adv}{\operatorname{Adv}}
\newcommand{\dist}{\operatorname{dist}}
\newcommand{\Par}{\mathrm{Par}}
\newcommand{\Dis}{\mathrm{Dis}}

\renewcommand{\>}{\rangle}
\newcommand{\<}{\langle}
\renewcommand{\ket}[1]{|#1\>}
\renewcommand{\bra}[1]{\<#1|}
\renewcommand{\braket}[2]{\<#1|#2\>}
\renewcommand{\proj}[1]{|#1\>\<#1|}

\newtheorem{theorem}{Theorem}[section]
\newtheorem{lemma}[theorem]{Lemma}
\newtheorem{proposition}[theorem]{Proposition}
\newtheorem{corollary}[theorem]{Corollary}
\newtheorem{conjecture}[theorem]{Conjecture}
\theoremstyle{definition}
\newtheorem{definition}[theorem]{Definition}
\theoremstyle{remark}
\newtheorem{remark}[theorem]{Remark}

\title{Quantum Algorithm for Identifying Hidden Graphs:\\Spectral Theory and Numerical Evidence}
\author{Pawe{\l} Wocjan\\[4pt]
\textit{IBM Quantum, T.J.\ Watson Research Center,
Yorktown Heights, NY 10598, USA}}
\date{}

\begin{document}
\maketitle

\begin{abstract}
We give a quantum algorithm for a novel type of black-box problem: \emph{identifying} a hidden $d$-regular base graph~$G$ on $n$ vertices from oracle access to an obfuscated version of it, rather than traversing it.  From~$G$ we build the \emph{spired graph}~$G_\spire$ in three steps: each vertex is lifted into an exponentially large cluster, with the clusters of adjacent vertices joined by a random bipartite graph; each cluster is then crowned with a balanced spire; finally, all vertices are randomly relabelled to obfuscate the structure.  Specializing to~$G=K_2$ recovers the welded-trees graph.

Our quantum algorithm is conceptually simple: a continuous-time quantum walk on~$G_\spire$, followed by a single Hadamard test at a classically precomputed evolution time~$t^*$; the algorithm returns the candidate whose predicted amplitude is closest to the measurement.  Its design and analysis rest on a rigorous spectral theory: from the apex of any spire, the walk is automatically confined to a polynomial-dimensional invariant subspace on which it evolves under the adjacency matrix of a much simpler \emph{towered graph}~$G_\tower$; that matrix block-diagonalizes into $n$ independent tridiagonal systems of size $n$, each solved in closed form by a Chebyshev secular equation.  Efficient numerics enabled by this decomposition diagonalize~$G_\tower$ for any d-regular base graph~$G$, supplying~$t^*$ and the predicted return amplitudes that the algorithm needs.

Specifically, on the prism graphs~$Y_m$ versus the M\"obius ladders~$M_m$ (each on $n=2m$ vertices), the numerical study supports a precise conjecture that $\widetilde O(n^2/\log n)$ measurements at evolution time of order $m^2$ suffice to distinguish the two families; we have tested $4\le m\le 5121$ ($n$ up to $10242$).  By analogy with the welded-trees lower bounds, we further conjecture that any classical algorithm requires queries exponential in $n$.  Together these conjectures point to an exponential quantum speedup for the identification of an obfuscated base graph.
\end{abstract}


\section{Introduction}\label{sec:intro}

We introduce a novel type of problem: \emph{identifying} which graph is hidden, not merely traversing it. The algorithm is given an obfuscating oracle for a spired graph~$G_\spire$ built from a hidden $d$-regular base graph $G$ on $n$~vertices, together with a list of candidates $G_1,\ldots,G_r$, and must identify $G$. This requires extracting global spectral information about the hidden graph, qualitatively different from the vertex-to-vertex navigation that has been the focus of earlier black-box quantum-walk separations.

From a $d$-regular base graph $G$ on $n$ vertices, we build the spired graph $G_\spire$ in three steps. The construction depends on a height parameter $L\ge 1$ and a thickening parameter $c\ge 1$; we write $D=cd$. Each vertex of $G$ is replaced by a cluster of $D^L$ vertices, with the clusters of adjacent vertices joined by a random $c$-regular bipartite graph; each cluster is then crowned with a balanced $D$-ary spire of depth $L$; finally, all vertices are randomly relabelled, yielding the obfuscating oracle $\cO_G$. The height $L$ is the security parameter of the construction.

The quantum algorithm is conceptually simple: a single-time Hadamard test estimating the return amplitude of a quantum walk at a classically precomputed optimal time~$t^*$, chosen to maximise the cross-candidate distinguishability of the predicted amplitudes. The algorithm returns the candidate whose predicted amplitude is closest to the measurement.  Its design and analysis rest on a rigorous spectral theory. The walk evolves under the adjacency matrix of $G_\spire$, an exponentially large Hamiltonian accessed only through the labelling oracle in superposition; from the apex, this unitary dynamics is automatically confined to a polynomial-dimensional invariant subspace whose effective Hamiltonian is the much simpler \emph{towered graph} $G_\tower$. The classical precomputation then diagonalizes $G_\tower$ in closed form via a Chebyshev secular equation, yielding $t^*$ and the predicted amplitudes for each candidate (quantities computable from the candidate graphs alone, without any oracle access).

For our quantum evaluation we work on the prism graphs~$Y_m$ versus the M\"obius ladders~$M_m$ (each $3$-regular on $n=2m$ vertices) with height $L=n-1$. The numerical study supports a precise efficiency conjecture: $\widetilde O(n^2/\log n)$ measurements at evolution time of order $m^2$ suffice to distinguish the two families; we have tested $4\le m\le 5121$ ($n$ up to $10242$).

Our results open a new direction within a recent line of exponential black-box separations between quantum and classical computation: identification, not traversal. The welded trees problem of Childs et~al.~\cite{childs03exponential} was the original instance: a quantum walk traverses the structure in polynomial time, while any classical algorithm requires exponential queries. Balasubramanian, Li, and Harrow~\cite{balasubramanian23hierarchical} proved superpolynomial-to-exponential hitting-time speedups for quantum walks on random hierarchical graphs, and Li~\cite{li23pathfinding} gave exponential quantum speedups for pathfinding on welded-tree variants. In each of these settings the quantum algorithm \emph{reaches a destination} (an exit vertex or a hitting set) through a hidden structure. In our construction, specializing to~$G=K_2$ recovers the welded-trees graph itself.

For the classical-hardness conjecture we further specialize: we set $c=2$ with each random $2$-regular bipartite graph drawn as a uniformly random alternating Hamiltonian cycle (inheriting the welded-trees obfuscation mechanism), and place the distinguished vertex at $u=((m-1)/2,\,0)$ on the outer rail (Section~\ref{sec:placement}) to keep the differing edges of the prism vs M\"obius pair (Section~\ref{sec:Y_M}) maximally far from the algorithm's starting point. By analogy with the welded-trees lower bounds, we conjecture that any classical algorithm requires queries exponential in $n$. Together with the efficiency conjecture, this points to an exponential quantum speedup for the identification of an obfuscated base graph.

The paper is organized as follows. Section~\ref{sec:construction} defines the spired graph $G_\spire$ and the obfuscating oracle. Sections~\ref{sec:effective_krylov}--\ref{sec:root_weight} develop the spectral theory of the towered graph $G_\tower$: a polynomial-dimensional invariant subspace of the quantum walk is identified (Section~\ref{sec:effective_krylov}); the prism and M\"obius-ladder base graphs are described (Section~\ref{sec:families}); the $n^2$-dimensional eigenvalue problem on $G_\tower$ is decomposed into $n$ independent tridiagonal systems (Section~\ref{sec:spectral}), each solved in closed form by a Chebyshev secular equation (Section~\ref{sec:secular}); top weights and the return amplitude follow (Section~\ref{sec:root_weight}). Section~\ref{sec:algorithm} presents the cross-graph spectral test and its measurement-budget theorem, with the efficient classical precomputation described in Section~\ref{sec:efficient}. Sections~\ref{sec:numerics} and~\ref{sec:hardness} present the numerical evidence and the classical-hardness conjecture. Section~\ref{sec:conclusions} discusses open problems.


\section{Construction of the spired graph}\label{sec:construction}

The construction takes three inputs: a simple, connected, $d$-regular graph $G=(V, E)$ with $n=|V|$ vertices and adjacency matrix $A\in\R^{n\times n}$; a \emph{height parameter} $L\in\N$, $L\ge 1$, the height of the spire over each base-graph vertex; and a \emph{distinguished vertex} $u\in V$, the unique vertex of~$G$ at which the algorithm is given an entry point. The construction also depends on a positive integer $c\ge 1$, the \emph{thickening parameter}, and we set $D=cd$. We write $\Sigma=\{1,2,\ldots,D\}$ for an alphabet of size~$D$. Throughout the paper, $u$ denotes the distinguished vertex of $G$, while $v$ and $w$ are reserved for arbitrary base-graph vertices; we write edges as $\{v,w\}\in E$.

The construction proceeds in three steps:
\begin{itemize}
\item the \emph{lifted graph} $G_\lift$, a thickening of~$G$;
\item the \emph{spired graph} $G_\spire$, obtained from $G_\lift$ by crowning each lifted cluster with a $D$-ary spire;
\item the \emph{obfuscating oracle} $\cO_G$, which presents $G_\spire$ to any algorithm only through random labels.
\end{itemize}

\subsection{The lifted graph}\label{sec:lifted}

For each $v\in V$, write $S_v^{(L)} = \{(v, x) : x\in\Sigma^L\}$ for the \emph{cluster} of $D^L$ length-$L$ strings indexed at~$v$. The lifted graph replaces each vertex of $G$ with its cluster, and each edge with a random $c$-regular bipartite graph between the corresponding clusters.

\begin{definition}[Random $c$-regular bipartite graph $C_e$]\label{def:Ce}
For an edge $e=\{v, w\}\in E$, let $\Omega_e$ be the set of all $c$-regular bipartite graphs on $S_w^{(L)}\uplus S_v^{(L)}$. Sample $C_e$ uniformly at random from $\Omega_e$; this contributes $cD^L$ edges per original edge of~$G$. We write $\mathbf{C}=(C_e)_{e\in E}$ for the full family of sampled graphs.
\end{definition}

\begin{definition}[Lifted graph $G_\lift$]\label{def:lift}
The \emph{lifted graph} $G_\lift = (V_\lift, E_\lift)$ has vertex set $V_\lift = \biguplus_{v\in V} S_v^{(L)}$ and edge set $E_\lift = \biguplus_{e\in E} E(C_e)$. It is $D$-regular.
\end{definition}

The cluster $S_v^{(L)}$ should be thought of as a \emph{thickened} version of vertex~$v$: the $d$ original neighbours of $v$ in $G$ contribute $d$ random $c$-regular bipartite graphs attached at the cluster, each giving every vertex in $S_v^{(L)}$ exactly $c$ neighbours in the partner cluster, for a total of $D=cd$ neighbours per cluster vertex.

The natural projection $\pi:V_\lift\to V$, mapping every vertex $(v,x)\in S_v^{(L)}$ to~$v$, is a graph homomorphism. At $c=1$, $C_e$ is a uniformly random perfect matching $S_v^{(L)}\leftrightarrow S_w^{(L)}$ and $\pi$ becomes a local isomorphism, exhibiting $G_\lift$ as a $d^L$-fold random covering of $G$ in the sense of algebraic graph theory~\cite{godsil2001algebraic}. For $c\ge 2$, $\pi$ fails to be a local isomorphism (each cluster vertex has $c$ neighbours in $S_w^{(L)}$ while $v$ has only one neighbour~$w$), and $G_\lift$ is more accurately described as a \emph{thickening} of $G$ in which each base-graph edge~$e$ is replaced by the $c$-regular bipartite graph $C_e$ on $2D^L$ vertices.

\subsection{Adding the spires}\label{sec:spires}

To obtain the spired graph $G_\spire$, we crown each cluster $S_v^{(L)}$ of $G_\lift$ with a spire defined using the prefix structure of strings over~$\Sigma$.

\begin{definition}[Spire $G_v$]\label{def:spire}
The spire $G_v = (S_v, E_v)$ at vertex $v$ is the graph whose vertex set is the disjoint union of $L+1$ level sets,
\[
S_v \;=\; \biguplus_{\ell=0}^L S_v^{(\ell)},
\qquad S_v^{(\ell)} \;=\; \{(v, s) : s\in\Sigma^\ell\},
\]
where $\Sigma^0=\{\epsilon\}$ consists of the empty string. The level $\ell=0$ is a singleton, $S_v^{(0)}=\{(v,\epsilon)\}$, and we call $a_v=(v,\epsilon)$ the \emph{apex} of the spire. The level $\ell=L$ coincides with the cluster $S_v^{(L)}$ already present in $G_\lift$, and we call this set the \emph{foundation} of the spire. The edge set $E_v$ is given by the \emph{prefix property}: for each level $\ell=0,\ldots,L-1$, each vertex $(v,s)\in S_v^{(\ell)}$ is connected to the $D$ vertices $(v, s\cdot\alpha)\in S_v^{(\ell+1)}$ ($\alpha\in\Sigma$).
\end{definition}

As a graph, $G_v$ is a perfectly balanced $D$-ary tree of depth~$L$, but \emph{inverted}: the apex $a_v$ is the unique root and protrudes outward as the only addressable entry point, while the foundation $S_v^{(L)}$ is the set of leaves and forms the base on which the spire stands.

\begin{definition}[Spired graph $G_\spire$]\label{def:spired}
The \emph{spired graph} $G_\spire$ is obtained by crowning each cluster in $G_\lift$ with its corresponding spire:
\[
G_\spire \;=\; \biggl(\biguplus_{v\in V} S_v,\;\biguplus_{v\in V} E_v\;\uplus\;E_\lift\biggr).
\]
\end{definition}

The spire acts as a tool that hides the foundation $S_v^{(L)}$ behind its single apex, addressable from outside as the only entry point of the spire. Internal vertices of the spire each have degree $D+1$ inside the spire alone, while the apex has degree~$D$ since it has no parent; foundation vertices have degree $D+1$ in $G_\spire$, contributed by one parent edge upward into the spire and the $D$ edges from the $d$ alternating cycles incident at $S_v^{(L)}$.

The height $L$ controls the size of the spired graph:
\begin{align*}
|V_\spire| = n\sum_{\ell=0}^L D^\ell = n\,\frac{D^{L+1}-1}{D-1},
\end{align*}
exponential in~$L$.

The construction is illustrated in Figures~\ref{fig:towered} and~\ref{fig:spired}. Figure~\ref{fig:spired} draws the random cycle only between the bottom and right spires, since including all random cycles would make the figure too crowded. We use $L$ to denote the length of the attached paths in the towered-graph figure.

\subsection{The obfuscating oracle}\label{sec:oracle}

The spired graph $G_\spire$ is not given explicitly. Instead, the algorithm receives an obfuscating oracle $\cO_G$ that exposes $G_\spire$ only through random labels.

\begin{definition}[Obfuscating oracle $\cO_G$]\label{def:oracle}
Let $\cX=\{0,1\}^a$ be a set of $a$-bit labels, with $a = 2 \lceil \log_2 |V_\spire| \rceil = O(L \log D + \log n)$, so that the probability that a uniformly random $a$-bit string falls inside the image $\iota(V_\spire)$ is at most $|V_\spire|/2^a \le 1/|V_\spire|$, exponentially small in $L$. Sample an injective labelling $\iota:V_\spire\to\cX$ uniformly at random, independently of $\mathbf{C}$. Given a label $x=\iota(w)\in\iota(V_\spire)$, the oracle returns the \emph{set}
\[
\cO_G(x) \;=\; \bigl\{\iota(y) : y\in N_{G_\spire}(w)\bigr\}
\]
of labels of neighbours of~$w$. On any label outside $\iota(V_\spire)$ the oracle returns $\bot$. The algorithm receives the seed label $\iota(a_{u})$ of the apex of the spire over the distinguished vertex as input.
\end{definition}

Three points deserve emphasis. First, the labelling $\iota$ is injective, so distinct vertices have distinct labels and the response $\cO_G(x)$ is genuinely a set rather than a multiset. Second, the oracle exposes no further positional or type information beyond the obvious degree distinction at the apices: every vertex has degree $D+1$ except the $n$ apices, which have degree~$D$. The algorithm cannot tell from a label whether the corresponding vertex is an internal spire vertex or a foundation vertex, nor which spire or which foundation a vertex belongs to. Third, the structural information about $G$ is doubly hidden, by the random sampling of $\mathbf{C}$ and by the uniformly random labelling $\iota$: in particular, $\iota$ destroys the prefix structure that would otherwise identify the level $\ell$ of a spire vertex from its $\Sigma$-string coordinates.

Because the labelling $\iota$ randomises away every distinguishing feature of $V_\spire$ except the apex degree, the only structural information available to the algorithm is the value of $L$ (and the apex count $n$). The height $L$ thus functions as the natural security parameter of the construction: any classical algorithm that recovers~$G$ from $\cO_G$ is conjectured to require a number of queries exponential in~$L$, with a base that grows with the thickening~$c$ (Section~\ref{sec:hardness}).

In the security analysis (this section and Section~\ref{sec:hardness}), we specialize to $c=2$ and require each $C_e$ to be a uniformly random \emph{alternating Hamiltonian cycle} on $S_v^{(L)}\uplus S_w^{(L)}$, a strict subset of the $2$-regular bipartite graphs (which generically split into a disjoint union of even cycles). At $G=K_2$, $d=1$, the spired graph then specializes exactly to the welded-trees graph of Childs et al.~\cite{childs03exponential}: two binary trees joined at their leaves by an alternating Hamiltonian cycle.

This conservative choice inherits the welded-trees obfuscation mechanism (random alternating cycles plus random labelling), whose classical opacity Childs et al.\ and Fenner--Zhang~\cite{fenner03lowerbound} established for the traversal version of the problem; the formal comparison and the resulting classical-hardness conjecture are taken up in Section~\ref{sec:hardness}.

In the quantum setting, the oracle is implemented as a unitary that maps $\ket{x}\ket{0}\mapsto\ket{x}\ket{\cO_G(x)}$ in superposition. This enables simulation of the spired-graph Hamiltonian $H_\spire=\frac{1}{c}A_\spire$ as a sparse Hamiltonian (maximum degree $D{+}1$), using standard sparse Hamiltonian simulation techniques~\cite{berry15hamiltonian}. The simulation cost per unit time is polylogarithmic in the precision and polynomial in the degree and the norm.

\begin{figure}[ht]
  \centering
  \includegraphics[width=0.8\textwidth]{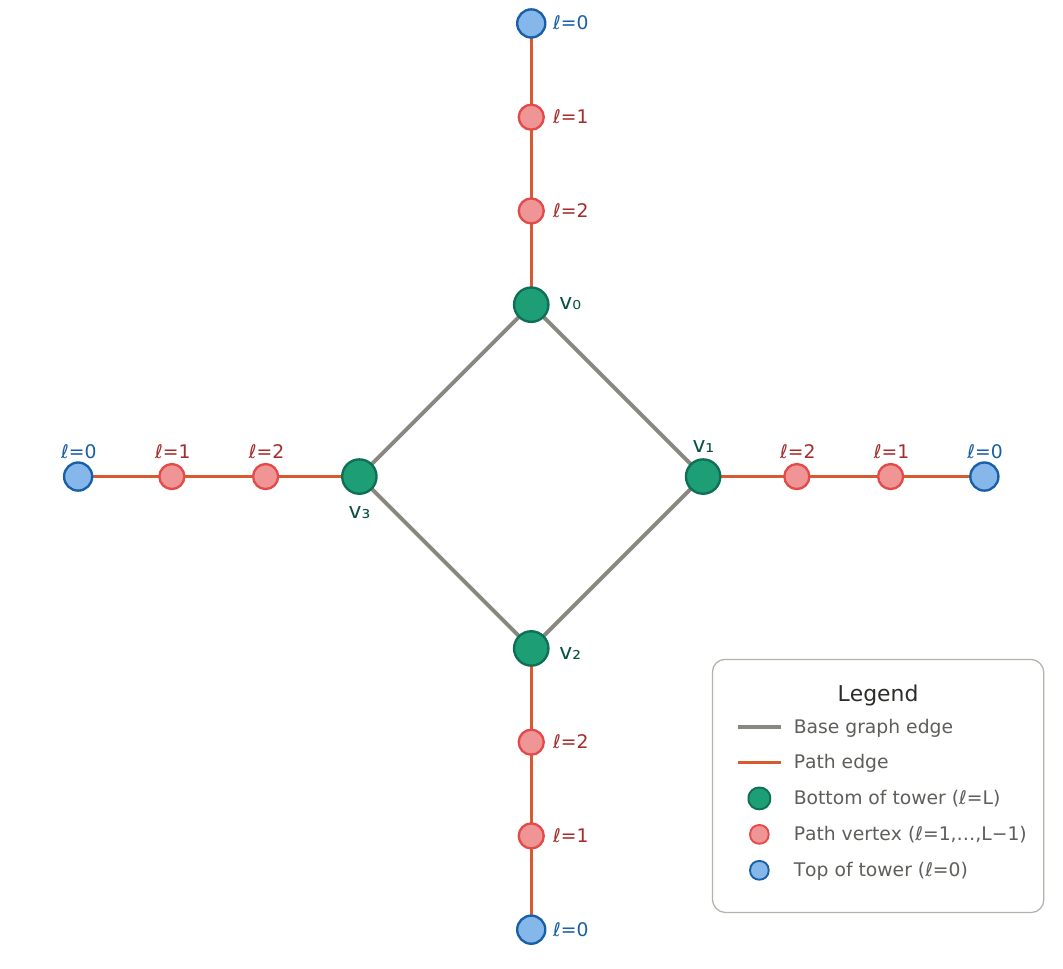}
  \caption{The towered graph $G_\tower$ for the cycle $C_4$ as the base graph, with path length $L=3$.
  The four vertices ($v_0,v_1,v_2,v_3$) are the bottoms of the towers ($\ell=L$), each identified with a vertex of the base cycle.
  The red vertices are interior path vertices ($\ell=1,\ldots,L{-}1$), and the blue vertices are the tops of the towers ($\ell=0$), where the quantum walk starts.}
  \label{fig:towered}
\end{figure}
 
\begin{figure}[ht]
  \centering
  \includegraphics[width=\textwidth]{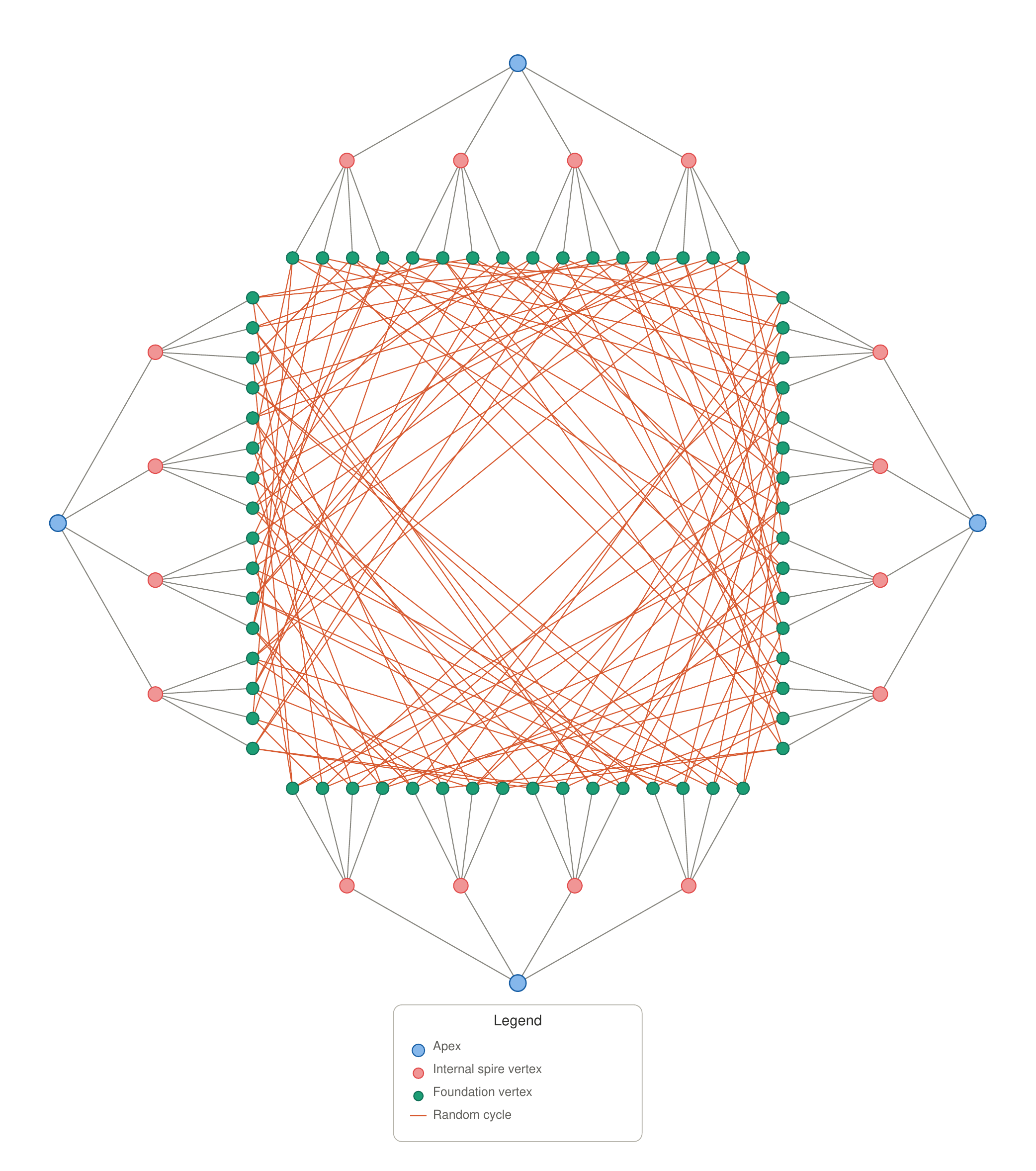}
  \caption{The spired graph $G_\spire$ for the cycle $C_4$ as the base graph, with degree $d=2$, thickening $c=2$, $D=cd=4$, and height $L=2$ (so each spire has one apex, four internal vertices, and a foundation of sixteen vertices).
  The foundations $S_v^{(L)}$ (green) face inward; the apices (blue) point outward.
  A random alternating Hamiltonian cycle $C_e$ (the $c=2$ specialization of Definition~\ref{def:Ce}; see Section~\ref{sec:oracle}) is attached to each of the four edges of $C_4$; all four are drawn in a single colour, so the foundation-to-foundation edges form a single visual tangle rather than four legible cycles.}
  \label{fig:spired}
\end{figure}

\FloatBarrier


\section{The effective Krylov Hamiltonian}\label{sec:effective_krylov}

The spired graph $G_\spire$ has exponentially many vertices, making direct analysis intractable. However, the quantum walk started from an apex evolves within a \emph{polynomially-dimensional invariant subspace}, whose structure we now describe.

\subsection{Level states and the Krylov subspace}\label{sec:level_states}

Define the Hamiltonian of the spired graph as $H_\spire = \frac{1}{c}A_\spire$. Throughout this section we work with the starting state $\ket{a_{u}} = \ket{(u,\epsilon)}$, where $u\in V$ is the distinguished vertex of Section~\ref{sec:construction} and $a_{u}$ is the apex of the spire $G_u$.

\begin{definition}[Level states]\label{def:level_state}
For each $v\in V$ and $\ell\in\{0,\ldots,L\}$, the \emph{level state} is the uniform superposition over all vertices at level~$\ell$ of the spire:
\[
\ket{S_v^{(\ell)}} = \frac{1}{\sqrt{D^\ell}}\sum_{x\in\Sigma^\ell}\ket{(v,x)}.
\]
In particular, $\ket{S_v^{(0)}}=\ket{(v,\epsilon)}$ is the apex.
\end{definition}

The level states span an $n(L{+}1)$-dimensional subspace of the exponentially large Hilbert space. The crucial observation is that this subspace is invariant under $H_\spire$.

\begin{definition}[Krylov subspace]\label{def:krylov}
The \emph{Krylov subspace} of the starting state $\ket{(u,\epsilon)}$ is
\[
\cK = \mathrm{span}\{A_\spire^j\ket{(u,\epsilon)} : j=0,1,2,\ldots\}.
\]
\end{definition}

\subsection{Action of the Hamiltonian on level states}\label{sec:action}

The adjacency matrix of $G_\spire$ decomposes as $A_\spire = A_{\lift}+\sum_{v\in V}A_v$, where $A_{\lift}$ is the adjacency matrix of the lifted graph and $A_v$ is the adjacency matrix of the spire $G_v$. Due to the regularity of the lifted graph and the perfect balance of the spires, the level states transform in a particularly simple way.

\begin{lemma}[Action on level states]\label{lem:action}
For all $v\in V$ and $\ell\in\{0,\ldots,L\}$:
\begin{align}
A_{\lift}\ket{S_v^{(L)}} &= c\sum_{\{v, w\}\in E}\ket{S_w^{(L)}}, &
A_{\lift}\ket{S_v^{(\ell)}} &= 0 \quad\text{for }\ell<L,
\label{eq:Alift_action}\\
A_v\ket{S_v^{(\ell)}} &= \sqrt{D}\bigl(\ket{S_v^{(\ell-1)}}+\ket{S_v^{(\ell+1)}}\bigr), &
A_v\ket{S_w^{(\ell)}} &= 0 \quad\text{for }v\neq w,
\label{eq:Aspire_action}
\end{align}
where we set $\ket{S_v^{(-1)}}=\ket{S_v^{(L+1)}}=0$.
\end{lemma}

\begin{proof}
For the lifted graph: each vertex in $S_v^{(L)}$ has exactly $c$ neighbours in $S_w^{(L)}$ for each edge $\{v, w\}\in E$ (from the random $c$-regular bipartite graph $C_e$ between $S_v^{(L)}$ and $S_w^{(L)}$). The uniform superposition $\ket{S_v^{(L)}}$ therefore maps to $c\sum_w\ket{S_w^{(L)}}$, where the factor $c$ arises from the $c$-regularity of $C_e$. For $\ell<L$, the vertices in $S_v^{(\ell)}$ belong to the tree interior and have no lifted-graph edges, giving zero.

For the tree adjacency: a vertex $(v,x)\in S_v^{(\ell)}$ with $0<\ell<L$ has one parent in $S_v^{(\ell-1)}$ and $D$ children in $S_v^{(\ell+1)}$. The normalization factors in the level states produce: the parent contributes $D^{-(\ell-1)/2}$ (from $\ket{S_v^{(\ell-1)}}$) with a factor $D^{\ell/2}$ (from dualizing $\bra{S_v^{(\ell)}}$), giving $\sqrt{D}$; the children similarly contribute $\sqrt{D}$. The boundary cases $\ell=0$ (no parent) and $\ell=L$ (no children) give the convention $\ket{S_v^{(-1)}}=\ket{S_v^{(L+1)}}=0$.
\end{proof}

\subsection{The towered graph}\label{sec:dec}

The structural data of the restriction (which level states couple to which, with what weights) collects naturally into a small graph that we call the towered graph.

\begin{definition}[Towered graph $G_\tower$]\label{def:tower}
The \emph{towered graph} $G_\tower$ has:
\begin{itemize}
\item Vertex set $V_\tower = \{(v,\ell): v\in V,\,\ell=0,\ldots,L\}$.
\item \emph{Top:} $(v,0)$, the top of the tower at $v$.
\item \emph{Bottom:} $(v,L)$, identified with the base-graph vertex $v$.
\item Tower edges $\{(v,\ell),(v,\ell{+}1)\}$ with weight $\gamma=\sqrt{D}/c=\sqrt{d/c}$ for $\ell=0,\ldots,L{-}1$.
\item Base-graph edges $\{(v,L),(w,L)\}$ with weight $1$ for $\{v, w\}\in E$.
\end{itemize}
At $c=2$, $\gamma=\sqrt{d/2}$.
\end{definition}

Write $A_\tower$ for the (weighted) adjacency matrix of $G_\tower$.

\begin{theorem}[Restriction to the Krylov subspace]\label{thm:restriction}
The Krylov subspace $\cK$ is spanned by the level states $\{\ket{S_v^{(\ell)}}:v\in V,\,\ell=0,\ldots,L\}$ and has dimension $n(L{+}1)$. The restriction of $H_\spire=\frac{1}{c}A_\spire$ to $\cK$ (the \emph{effective Krylov Hamiltonian} $H_\spire|_\cK$) equals the towered-graph adjacency matrix:
\[
H_\spire|_\cK \;=\; A_\tower,
\]
where the identification of $\cK$ with $\C^{V_\tower}$ is via $\ket{S_v^{(\ell)}}\leftrightarrow\ket{(v,\ell)}$.
\end{theorem}

\begin{proof}
The level state $\ket{S_{u}^{(0)}}=\ket{(u,\epsilon)}$ is the starting state, so it belongs to $\cK$. By Lemma~\ref{lem:action}, applying $A_\spire$ to any level state produces a linear combination of level states, so $\cK\subseteq\mathrm{span}\{\ket{S_v^{(\ell)}}\}$. Since $\ket{(u,\epsilon)}\in\cK$ and the level states at all levels and vertices can be reached by iterated application (the tower edges connect adjacent levels, and the lifted graph connects different vertices at level~$L$), the Krylov subspace equals the full span.

For the effective Krylov Hamiltonian: using $H_\spire=\frac{1}{c}A_\spire = \frac{1}{c}(A_{\lift}+\sum_v A_v)$ and Lemma~\ref{lem:action}, the matrix elements in the level-state basis are:
\begin{align*}
\bra{S_v^{(\ell)}}H_\spire\ket{S_v^{(\ell\pm 1)}} &= \frac{\sqrt{D}}{c} = \gamma,\\
\bra{S_v^{(L)}}H_\spire\ket{S_w^{(L)}} &= \frac{c}{c}\cdot[\{v, w\}\in E] = [\{v, w\}\in E],
\end{align*}
and all other matrix elements vanish. The tower-edge weight may be rewritten as $\gamma=\sqrt{D}/c=\sqrt{cd}/c=\sqrt{d/c}$. These are exactly the matrix entries of $A_\tower$ in the basis $\{\ket{(v,\ell)}\}$.
\end{proof}

\begin{remark}[Terminology]\label{rem:terminology}
Two pairs of names refer to the corresponding positions of $G_\spire$ and $G_\tower$: on the spire side, the \emph{apex} $a_v=(v,\epsilon)$ (a single vertex of $G_v$) and the \emph{foundation} $S_v^{(L)}$ (the cluster of $D^L$ leaves of $G_v$); on the tower side, the \emph{top} $(v,0)$ and the \emph{bottom} $(v,L)\equiv v$. The path Hamiltonian $H^{(\mu)}$ of Section~\ref{sec:secular} acquires a top condition and a bottom condition at the tower endpoints. The architectural correspondence between the two pairs is taken up in Section~\ref{sec:architecture}.
\end{remark}

For our algorithm, we set $L=n-1$, so each spire has height $L$ and each path of $G_\tower$ has $L$ edges (and $n$ vertices), giving $|V_\tower|=n^2$.

\subsection{Architectural visualization}\label{sec:architecture}

Picture the base graph $G$ in a horizontal plane with a tower (a path of length $L$) pointing vertically upward over each vertex; this is $G_\tower$. Now picture the lifted graph $G_\lift$ in the horizontal plane with a spire (an inverted $D$-ary tree of depth $L$) pointing upward over each cluster; this is $G_\spire$. The towered graph is the simplified version of the spired graph: $G_\tower$ is the natural quotient of $G_\spire$ under the level-set partition, and Theorem~\ref{thm:restriction} records this quotient spectrally as $H_\spire|_\cK = A_\tower$.

The two graphs are related by the \emph{quotient projection} $\pi : G_\spire \to G_\tower$ induced by the equitable partition of $V_\spire$ into the level sets $\{S_v^{(\ell)} : v\in V,\,\ell=0,\dots,L\}$ (cf.~\cite{godsil2001algebraic}, ch.~9). Each spire projects to its tower; the apex $a_v$ to the top $(v,0)$; each foundation vertex (a leaf of the spire, equivalently a lifted copy of $v$) to the bottom $(v,L)$, which is the base-graph vertex $v$ itself. The full foundation $S_v^{(L)}$, of $D^L$ vertices in $V_\spire$, thus projects to the single vertex $(v,L) \in V_\tower$. In particular, the starting state $\ket{a_u}$ identifies under $\pi$ with the top of the tower at $u$, aligning with the welded-trees terminology of~\cite{childs03exponential}, where the walk starts at the tree root (our apex).

The tower-edge weight $\gamma=\sqrt{D}/c=\sqrt{d/c}$ is precisely the weighting that makes $\pi$ a spectrally-faithful quotient: the level-state restriction of Theorem~\ref{thm:restriction} realises this on the spectral side, identifying the level state $\ket{S_v^{(L)}}$ (the uniform superposition over the foundation $S_v^{(L)}$) with the basis vector $\ket{(v,L)}$ of $G_\tower$.

\smallskip
\begin{center}
\begin{tabular}{lll}
\toprule
                            & $G_\spire$                                 & $G_\tower$ \\
\midrule
Vertical structure at $v$   & spire (inverted $D$-ary tree)        & tower at $v$ (path of length $L$) \\
Top ($\ell = 0$)            & apex $a_v$                                 & top vertex $(v, 0)$ \\
Bottom ($\ell = L$)         & foundation $S_v^{(L)}$ ($D^L$ vertices)            & bottom vertex $(v, L) \equiv v$ \\
Layer at $\ell = L$         & lifted graph $G_\lift$                     & base graph $G$ \\
\bottomrule
\end{tabular}
\end{center}

\smallskip
$G_\tower$ specializes to the Godsil--McKay (GM) graph construction~\cite{godsil78new_graph_product} at $c=d$: there $\gamma=\sqrt{d/d}=1$, so the tower edges share the unit weight of the base-graph edges, and $G_\tower$ is the unweighted graph obtained by attaching a path of length $L$ to each vertex of $G$.


\section{Base graph families: prism and M\"obius ladder}\label{sec:families}

We study two families of $3$-regular vertex-transitive graphs on $n=2m$ vertices, $m\ge 3$. Both families share most of their edge set; they differ in only four edges, which we call the \emph{closing rail edges}. This shared structure makes them an unusually transparent test case: any distinguishability between the two (spectral, algorithmic, or query-complexity-theoretic) is fully attributable to the choice of closing edges.

\subsection{Pair coordinates and shared structure}\label{sec:pair_coords}

Both graphs $Y_m$ (the prism) and $M_m$ (the M\"obius ladder) have the common vertex set
\[
V \;=\; \{(k, b) : 0\le k\le m-1,\;b\in\{0, 1\}\},
\qquad |V| \;=\; 2m.
\]
We call $\{(k, 0)\}_{0\le k\le m-1}$ the \emph{outer rail} and $\{(k, 1)\}_{0\le k\le m-1}$ the \emph{inner rail}. The two graphs share the following \emph{common edges}:
\begin{itemize}
\item \emph{rungs}: $\{(k, 0), (k, 1)\}$ for $0\le k\le m-1$;
\item \emph{open rail edges}: $\{(k, b), (k{+}1, b)\}$ for $0\le k\le m-2$ and $b\in\{0, 1\}$.
\end{itemize}
This common edge set forms a $2{\times}m$ ladder graph: two parallel paths of $m$ vertices each, connected at every position by a rung. Each vertex on the outer or inner rail has at most three incident common edges; each of the four \emph{rail endpoints} $(0, b), (m{-}1, b)$ ($b\in\{0,1\}$) has only two common edges and is therefore one rail edge short of degree~$3$.

\subsection{The prism and the M\"obius ladder}\label{sec:Y_M}

The two graphs differ only in how the rail endpoints are joined to close the ladder into a closed band:
\begin{align*}
E(Y_m) \setminus E(M_m)
   &= \bigl\{\{(m{-}1, 0), (0, 0)\},\ \{(m{-}1, 1), (0, 1)\}\bigr\}
   &&\text{(parallel)}, \\
E(M_m) \setminus E(Y_m)
   &= \bigl\{\{(m{-}1, 0), (0, 1)\},\ \{(m{-}1, 1), (0, 0)\}\bigr\}
   &&\text{(twisted)}.
\end{align*}
We call the four edges in $E(Y_m)\,\triangle\,E(M_m)$ the \emph{closing rail edges}, and the four-vertex set $\{(0, b), (m{-}1, b) : b\in\{0,1\}\}$ together with these four edges the \emph{differing $4$-cycle}, informally \emph{the twist}. In drawings (Figures~\ref{fig:prism_mobius_7} and~\ref{fig:prism_mobius_6}) we place the twist on the right.

\paragraph{Eigenvalues.} The prism $Y_m=C_m\square K_2$ is the Cartesian product of an $m$-cycle and an edge, so its eigenvalues are
\begin{equation}\label{eq:prism_evals}
\mu_k^{\pm} \;=\; 2\cos\!\left(\frac{2\pi k}{m}\right) \pm 1, \qquad k=0,1,\ldots,m-1.
\end{equation}
The M\"obius ladder $M_m$ is isomorphic to the cycle $C_{2m}$ with $m$ additional ``rung'' edges joining vertex~$i$ to vertex $i{+}m$ for $i=0,\ldots,m{-}1$, with eigenvalues
\begin{equation}\label{eq:moebius_evals}
\mu_k \;=\; 2\cos\!\left(\frac{2\pi k}{2m}\right) + (-1)^k, \qquad k=0,1,\ldots,2m-1.
\end{equation}
Both formulas follow from the standard spectral theory of vertex-transitive graphs.

\paragraph{Unified form.} The two formulas can be brought into a common form that exposes which eigenvalues coincide between $Y_m$ and $M_m$ and which differ. Writing $\theta_j := 2\pi j/m$ for the $m$ cycle-of-length-$m$ angles, the M\"obius formula~\eqref{eq:moebius_evals} splits by parity of $k$ ($k=2j$ even, $k=2j+1$ odd, $j=0,\ldots,m{-}1$) into two branches of $m$ eigenvalues each, matching the $\pm 1$ branches of the prism formula~\eqref{eq:prism_evals}.

\begin{proposition}[Unified spectral form]\label{prop:unified_evals}
For $j = 0, 1, \ldots, m-1$, the eigenvalues of $Y_m$ and $M_m$ are
\begin{align*}
Y_m{:} \quad & \mu_j^{+} = 2\cos\theta_j + 1, \qquad \mu_j^{-} = 2\cos\theta_j - 1, \\
M_m{:} \quad & \mu_j^{+} = 2\cos\theta_j + 1, \qquad \mu_j^{-} = 2\cos(\theta_j + \pi/m) - 1.
\end{align*}
The $+1$ branches coincide between the two graphs; the $-1$ branches differ by the half-step phase shift $\theta_j \mapsto \theta_j + \pi/m$ in the cosine argument.
\end{proposition}

\begin{proof}
For $Y_m$, the Cartesian product $Y_m = C_m \square K_2$ has eigenvalues that are sums of cycle and edge eigenvalues, $2\cos\theta_j \pm 1$. For $M_m$, the Cayley graph of $\mathbb{Z}_{2m}$ with generators $\{\pm 1, m\}$ has eigenvalues $\chi_k(1) + \chi_k(-1) + \chi_k(m) = 2\cos(\pi k/m) + (-1)^k$ at character $\chi_k(j) = e^{\pi i k j/m}$, $k = 0, \ldots, 2m{-}1$. Splitting by parity, $k=2j$ gives $\mu^+_j = 2\cos\theta_j + 1$, and $k=2j+1$ gives $\mu^-_j = 2\cos(\pi(2j+1)/m) - 1 = 2\cos(\theta_j + \pi/m) - 1$.
\end{proof}

\begin{remark}[Spectral reading of bipartiteness]\label{rem:bipartite_spectral}
A $3$-regular graph is bipartite iff $-3$ lies in its spectrum. By Proposition~\ref{prop:unified_evals}, $-3$ appears on the $-1$ branch precisely when the cosine argument equals $\pi$: for $Y_m$ this requires $\theta_j = \pi$ (i.e.\ $j = m/2$, hence $m$ even); for $M_m$ this requires $\theta_j + \pi/m = \pi$ (i.e.\ $2j + 1 = m$, hence $m$ odd). This reproduces the alternating bipartiteness of Proposition~\ref{prop:bipartite} directly from the spectrum.
\end{remark}

\begin{figure}[ht]
\centering
\begin{tikzpicture}[xscale=1.7, yscale=1.0]
  \begin{scope}
    \node[font=\small, anchor=south west] at (-3.6, 0.55) {\bfseries $+1$ branch (identical)};
    \draw[->, thick] (-3.6, 0) -- (3.6, 0);
    \foreach \x in {-3, -2, -1, 0, 1, 2, 3} {
      \draw (\x, -0.06) -- (\x, 0.06);
      \node[below, font=\small] at (\x, -0.55) {$\x$};
    }
    \foreach \xM in {3, 2.247, 2.247, 0.555, 0.555, -0.802, -0.802}
      \fill[red!75!black] (\xM, 0.25) circle (1.8pt);
    \foreach \xY in {3, 2.247, 2.247, 0.555, 0.555, -0.802, -0.802}
      \fill[blue!75!black] (\xY, -0.25) circle (1.8pt);
  \end{scope}
  \begin{scope}[yshift=-2.5cm]
    \node[font=\small, anchor=south west] at (-3.6, 0.55) {\bfseries $-1$ branch (M\"obius shifted by $\pi/m$)};
    \draw[->, thick] (-3.6, 0) -- (3.6, 0);
    \foreach \x in {-3, -2, -1, 0, 1, 2, 3} {
      \draw (\x, -0.06) -- (\x, 0.06);
      \node[below, font=\small] at (\x, -0.55) {$\x$};
    }
    \foreach \xM in {-3, -2.247, -2.247, -0.555, -0.555, 0.802, 0.802}
      \fill[red!75!black] (\xM, 0.25) circle (1.8pt);
    \foreach \xY in {1, 0.247, 0.247, -1.445, -1.445, -2.802, -2.802}
      \fill[blue!75!black] (\xY, -0.25) circle (1.8pt);
  \end{scope}
\end{tikzpicture}
\caption{Eigenvalues of $Y_7$ (blue, below the axis) and $M_7$ (red, above the axis), split by the $\pm 1$ branches of Proposition~\ref{prop:unified_evals}.  Top: on the $+1$ branch the two spectra coincide, so each $M_7$ marker has a $Y_7$ marker directly below it.  Bottom: on the $-1$ branch the $M_7$ markers are shifted by $\pi/m = \pi/7$ in the cosine argument relative to $Y_7$, and the leftmost $M_7$ eigenvalue reaches $-3$ at $\theta_3 + \pi/7 = \pi$, witnessing that $M_7$ is bipartite (Corollary~\ref{cor:exactly_one_bipartite}).  Coincident markers within a single graph correspond to the multiplicity-$2$ degeneracy from paired indices $j$ and $m{-}j$.}
\label{fig:eigenvalues_unified}
\end{figure}
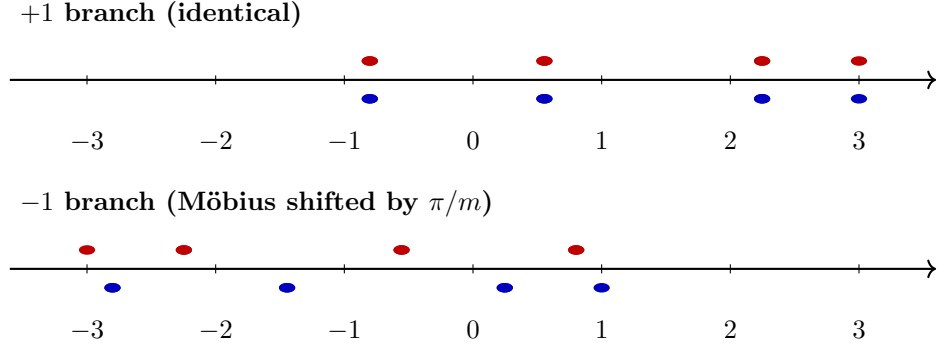

\paragraph{Vertex-transitivity.} Both $Y_m$ and $M_m$ are vertex-transitive. For our quantum algorithm this matters because the return amplitude $f_G(t) = \bra{v,0}\exp(-iA_\tower t)\ket{v,0}$ is then the same for every starting top vertex $(v, 0)\in V_\tower$; equivalently, the choice of distinguished vertex $u\in V$ does not affect the quantum signal. The classical hardness conjecture in Section~\ref{sec:hardness} will, however, exploit a particular choice of $u$: vertex transitivity gives us the freedom to place $u$ anywhere, so we may as well place it where the classical adversary is worst off.

\subsection{Bipartiteness}\label{sec:bipartite}

Define the parity colouring $c: V\to\{0, 1\}$ by
\[
c(k, b) \;=\; (k + b) \bmod 2.
\]

\begin{proposition}[Alternating bipartiteness]\label{prop:bipartite}
$Y_m$ is bipartite if and only if $m$ is even, and $M_m$ is bipartite if and only if $m$ is odd. In both bipartite cases, a $2$-colouring witness is the parity colouring~$c$.
\end{proposition}

\begin{proof}
The colouring $c$ assigns differing parities to the two endpoints of every rung and every open rail edge. The closing rail edges contribute as follows.
\begin{itemize}
\item In $Y_m$, the parallel closing edge $\{(m{-}1, b), (0, b)\}$ has parity difference $(m-1+b)-b = m-1$; this is odd iff $m$ is even, so $c$ is a proper $2$-colouring iff $m$ is even.
\item In $M_m$, the twisted closing edge $\{(m{-}1, b), (0, 1{-}b)\}$ has parity difference $(m-1+b)-(1-b) = m-2+2b\equiv m\pmod 2$; this is odd iff $m$ is odd.
\end{itemize}
Conversely, in each failing case the graph contains an explicit odd cycle. For $Y_m$ with $m$ odd, the outer rail $(0,0)\to(1,0)\to\cdots\to(m{-}1,0)\to(0,0)$ is an odd $m$-cycle. For $M_m$ with $m$ even, the closed walk $(0,0)\to(1,0)\to\cdots\to(m{-}1,0)\to(0,1)\to(0,0)$ (using $m{-}1$ open rail edges, one twisted closing edge, and one rung) is an odd cycle of length $m+1$.
\end{proof}

\begin{corollary}[Bipartite member toggles between families]\label{cor:exactly_one_bipartite}
For every integer $m\ge 3$, exactly one of the graphs $Y_m, M_m$ is bipartite: $Y_m$ when $m$ is even, $M_m$ when $m$ is odd. In particular, the bipartite member of the pair $\{Y_m, M_m\}$ alternates between the two families as $m$ increments by one.
\end{corollary}

\begin{proof}
Immediate from Proposition~\ref{prop:bipartite}, since every integer is either even or odd but not both.
\end{proof}

\begin{remark}[Algorithmic scope]\label{rem:algorithm_scope}
Corollary~\ref{cor:exactly_one_bipartite} suggests a problem-specific shortcut: since a bipartite $d$-regular graph has spectrum symmetric about~$0$, the return amplitude of the bipartite member of $\{Y_m, M_m\}$ is real-valued and $\mathrm{Im}\,f \equiv 0$ on that side.  An imaginary-branch algorithm exploiting this single-graph signal is developed in the supplementary material.  The cross-graph algorithm of Section~\ref{sec:algorithm}, by contrast, uses the full signal $f_{Y_m} - f_{M_m}$ and does not exploit bipartiteness; we present it because it remains meaningful for pairs of candidate graphs in which neither is bipartite.
\end{remark}

\begin{remark}\label{rem:K33}
For $m=3$ the M\"obius ladder $M_3$ coincides with the complete bipartite graph $K_{3,3}$, in agreement with Proposition~\ref{prop:bipartite}.
\end{remark}

The numerical study of Section~\ref{sec:numerics} exploits this alternation, sampling both parity classes at every scale via the doubling pairs $(2^k, 2^k{+}1)$ and reporting independent scaling fits on the even-$m$ and odd-$m$ subsequences.

\subsection{Symmetric placement of the distinguished vertex}\label{sec:placement}

Where on the outer rail should we place the distinguished vertex $u$ to make the classical-distinguishing problem of Section~\ref{sec:hardness} as hard as possible? Intuitively, $u$ should sit as far as possible from the differing $4$-cycle, since a classical algorithm starting at the apex $a_{u}$ must first traverse spires and random alternating cycles before any of the differing closing edges can come into view. The shared common-edge structure of $Y_m$ and $M_m$ makes the question precise.

\begin{proposition}[Symmetric placement]\label{prop:placement}
Suppose $m$ is odd, $m\ge 3$, and set $u := ((m{-}1)/2, 0)$. Then in both $Y_m$ and~$M_m$:
\begin{enumerate}
\item[\textup{(a)}] the two paths from $u$ along the open outer rail going clockwise and counter-clockwise reach $(0, 0)$ and $(m{-}1, 0)$ in exactly $(m{-}1)/2$ open-rail steps each;
\item[\textup{(b)}] $\dist(u, (0, 0)) = \dist(u, (m{-}1, 0)) = (m{-}1)/2$, and these are the two closest vertices of the differing $4$-cycle to~$u$;
\item[\textup{(c)}] the inner-rail vertices of the differing $4$-cycle are slightly farther: $\dist(u, (0, 1)) = \dist(u, (m{-}1, 1)) = (m{+}1)/2$.
\end{enumerate}
For $m$ even, no choice of $u$ on the outer rail puts the two outer-rail closing-edge endpoints at equal graph distance; the most symmetric placement $u = (m/2, 0)$ has distances $m/2$ and $m/2 - 1$, asymmetric by one.
\end{proposition}

\begin{proof}[Sketch]
The two outer-rail paths from $u$ to the closing-edge endpoints use only open rail edges, which are common to $Y_m$ and $M_m$, so all distances stated above agree across the two graphs. For $m$ odd, the open outer rail consists of $m-1$ edges between $(0, 0)$ and $(m{-}1, 0)$, and $u = ((m{-}1)/2, 0)$ is at index distance $(m{-}1)/2$ from each endpoint. Any path detouring through the inner rail uses at least two extra rung edges, so the outer-rail path is shortest. For $m$ even, the open outer rail has $m-1$ edges, an odd number, so no midpoint exists.
\end{proof}

For our purposes the distinction between $m$ odd and $m$ even matters: with $m$ odd, the two sides of the graph (clockwise and counter-clockwise from $u$) are interchangeable, and the differing $4$-cycle sits at maximum graph distance from $u$ along the outer rail. We will adopt this odd-$m$ symmetric placement when stating the classical-hardness conjecture in Section~\ref{sec:hardness}. The numerical results in Sections~\ref{sec:numerics} are insensitive to the choice of $u$ by vertex-transitivity, and so are presented for general $m\ge 4$.

\begin{figure}[ht]
\centering
\includegraphics[width=0.42\textwidth]{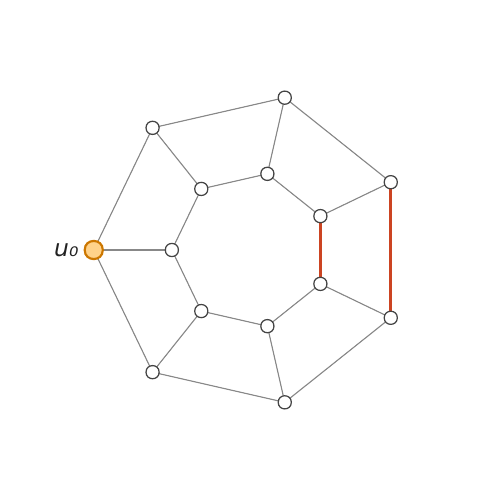}\hspace{0.04\textwidth}%
\includegraphics[width=0.42\textwidth]{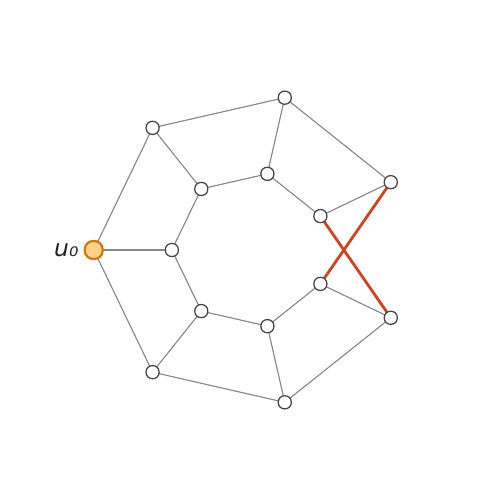}
\caption{The prism $Y_7$ (left) and the M\"obius ladder $M_7$ (right) in pair coordinates. Outer-rail vertices $(k, 0)$ lie on the larger circle, inner-rail vertices $(k, 1)$ on the smaller. The four edges in $E(Y_7)\,\triangle\,E(M_7)$ are drawn in red. In $Y_7$ they close each rail into its own $7$-cycle; in $M_7$ they cross to merge the two rails into a single $14$-cycle. With $m=7$ odd, the distinguished vertex $u = ((m{-}1)/2,\, 0) = (3, 0)$ sits at equal distance $(m{-}1)/2 = 3$ from both outer-rail closing-edge endpoints $(0, 0)$ and $(m{-}1, 0) = (6, 0)$ along the open outer rail. This symmetric placement is unavailable for $m$ even (Figure~\ref{fig:prism_mobius_6}).}
\label{fig:prism_mobius_7}
\end{figure}

\begin{figure}[ht]
\centering
\includegraphics[width=0.42\textwidth]{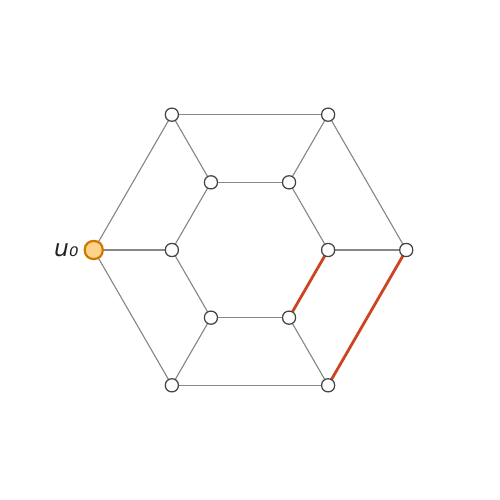}\hspace{0.04\textwidth}%
\includegraphics[width=0.42\textwidth]{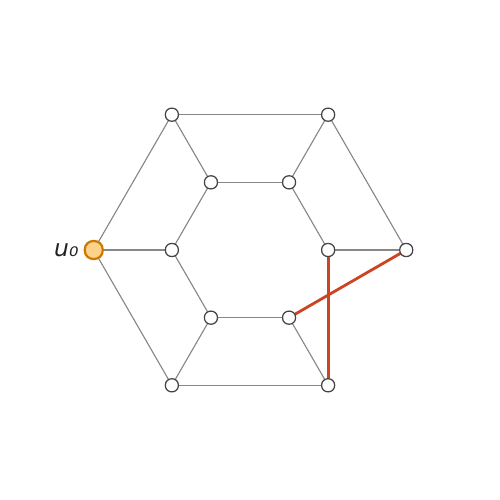}
\caption{The prism $Y_6$ (left) and the M\"obius ladder $M_6$ (right) in pair coordinates. As in Figure~\ref{fig:prism_mobius_7}, the four edges in $E(Y_6)\,\triangle\,E(M_6)$ are drawn in red. With $m=6$ even, the distinguished vertex $u = (m/2, 0) = (3, 0)$ shown on the left of each figure is at distance $3$ from $(0, 0)$ and at distance $2$ from $(m{-}1, 0) = (5, 0)$ along the open outer rail, an asymmetry of one. The symmetric placement of Proposition~\ref{prop:placement} is unavailable when $m$ is even.}
\label{fig:prism_mobius_6}
\end{figure}

\FloatBarrier


\section{Spectral theory of the towered graph}\label{sec:spectral}

We now develop a complete spectral decomposition of the towered graph $G_\tower$. The key insight is that $A_\tower$ admits a tensor-product decomposition that reduces the $n^2$-dimensional eigenvalue problem to $n$ independent problems of size $n$.

\subsection{Tensor product structure}\label{sec:tensor}

We identify the vertex $(v,\ell)$ of $G_\tower$ with the tensor product state $\ket{v}\otimes\ket{\ell}$, so the Hilbert space is $\cH=\C^n\otimes\C^{L+1}$, where $L=n-1$.

\begin{definition}[Path matrix]\label{def:path}
The weighted path matrix $P\in\R^{(L+1)\times(L+1)}$ is
\[
P = \gamma\sum_{\ell=0}^{L-1}\bigl(\ket{\ell}\bra{\ell{+}1}+\ket{\ell{+}1}\bra{\ell}\bigr).
\]
\end{definition}

\begin{proposition}[Tensor decomposition]\label{prop:tensor}
The adjacency matrix $A_\tower$ satisfies
\begin{equation}\label{eq:H}
A_\tower = A\otimes\proj{L} + I_n\otimes P,
\end{equation}
where $\proj{L}\in\R^{(L+1)\times(L+1)}$ is the projection onto the bottom $\ell=L$.
\end{proposition}

\begin{proof}
We verify that~\eqref{eq:H} reproduces the correct matrix elements. For a state $\ket{v,\ell}=\ket{v}\otimes\ket{\ell}$:

\emph{Case $1\le\ell\le L{-}1$ (interior of the path):} The only neighbors are $(v,\ell\pm 1)$, both connected by weight~$\gamma$. We compute:
$(A\otimes\proj{L})\ket{v,\ell}=A\ket{v}\otimes\braket{L}{\ell}\ket{L}=0$ since $\ell\neq L$. And $(I\otimes P)\ket{v,\ell}=\ket{v}\otimes(\gamma\ket{\ell{-}1}+\gamma\ket{\ell{+}1})$. Combined: $A_\tower\ket{v,\ell}=\gamma\ket{v,\ell{-}1}+\gamma\ket{v,\ell{+}1}$.~$\checkmark$

\emph{Case $\ell=0$ (top):} The only neighbor is $(v,1)$. We get $(A\otimes\proj{L})\ket{v,0}=0$ and $(I\otimes P)\ket{v,0}=\gamma\ket{v,1}$.~$\checkmark$

\emph{Case $\ell=L$ (bottom):} The neighbors are $(v,L{-}1)$ (weight $\gamma$) and all $(w,L)$ for $\{v, w\}\in E$ (weight~$1$). We get $(A\otimes\proj{L})\ket{v,L}=A\ket{v}\otimes\ket{L}=\sum_w A_{vw}\ket{w,L}$ and $(I\otimes P)\ket{v,L}=\gamma\ket{v,L{-}1}$.~$\checkmark$
\end{proof}

\subsection{Block diagonalization}\label{sec:block}

The tensor structure enables an exact block diagonalization.

\begin{lemma}[Commutativity]\label{lem:commute}
$[A_\tower,\;A\otimes I_{L+1}] = 0$.
\end{lemma}

\begin{proof}
We compute each term separately:
\begin{align*}
[A\otimes\proj{L},\;A\otimes I] &= (A^2-A^2)\otimes\proj{L} = 0,\\
[I\otimes P,\;A\otimes I] &= (A-A)\otimes P = 0.
\end{align*}
Both commutators vanish because $A$ commutes with itself and $I$ commutes with everything.
\end{proof}

Since $A_\tower$ commutes with $A\otimes I$, both operators can be simultaneously diagonalized. Let $A = \sum_{j=1}^n \mu_j\proj{\alpha_j}$ be the spectral decomposition of $A$, with eigenvalues $\mu_1\le\cdots\le\mu_n$ and orthonormal eigenvectors $\ket{\alpha_1},\ldots,\ket{\alpha_n}$.

\begin{definition}[Path Hamiltonian with boundary potential]\label{def:path_ham}
For each $\mu\in\R$, the \emph{path Hamiltonian with boundary potential~$\mu$} is the $(L{+}1)\times(L{+}1)$ tridiagonal matrix
\begin{equation}\label{eq:path_ham}
H^{(\mu)} = \mu\proj{L} + P =
\begin{pmatrix}
0 & \gamma \\
\gamma & 0 & \gamma \\
& \gamma & \ddots & \ddots \\
& & \ddots & 0 & \gamma \\
& & & \gamma & \mu
\end{pmatrix}.
\end{equation}
This is a path graph with uniform edge weight~$\gamma$ and a boundary potential $\mu$ at the bottom (position $L$, the last diagonal entry). The top (position $0$) has no potential.
\end{definition}

\begin{figure}[ht]
  \centering
  \definecolor{topfill}{HTML}{85B7EB}
  \definecolor{topstroke}{HTML}{185FA5}
  \definecolor{intfill}{HTML}{F09595}
  \definecolor{intstroke}{HTML}{E24B4A}
  \definecolor{botfill}{HTML}{1D9E75}
  \definecolor{botstroke}{HTML}{0F6E56}
  \definecolor{pathedge}{HTML}{D85A30}
  \definecolor{baseedge}{HTML}{888780}
  \begin{tikzpicture}[
    vertex/.style={circle, minimum size=11pt, inner sep=0pt, thick},
    apex/.style={vertex, fill=topfill, draw=topstroke},
    int/.style={vertex, fill=intfill, draw=intstroke},
    bot/.style={vertex, fill=botfill, draw=botstroke},
  ]
    \node[apex, label={below:$\ell=0$}] (l0) at (0,0) {};
    \node[int,  label={below:$\ell=1$}] (l1) at (2,0) {};
    \node[int,  label={below:$\ell=2$}] (l2) at (4,0) {};
    \node[bot,  label={below:$\ell=3$}] (l3) at (6,0) {};

    \draw[pathedge, thick] (l0) -- node[above] {$\gamma$} (l1);
    \draw[pathedge, thick] (l1) -- node[above] {$\gamma$} (l2);
    \draw[pathedge, thick] (l2) -- node[above] {$\gamma$} (l3);

    \draw[baseedge, thick] (l3) edge[out=35, in=-35, looseness=22, min distance=18pt]
          node[right] {$\mu$} (l3);
  \end{tikzpicture}
  \caption{The path Hamiltonian $H^{(\mu)}$ as a weighted graph: a path on $L+1$ vertices indexed by $\ell = 0, 1, \ldots, L$ (drawn for $L=3$), with uniform edge weight $\gamma$ and a single self-loop of weight $\mu$ at the bottom vertex $\ell=L$.
  The self-loop captures the boundary effect of the base-graph eigenvalue~$\mu$; by Proposition~\ref{prop:block} below, $A_\tower$ decomposes as $\bigoplus_{j=1}^n H^{(\mu_j)}$, a direct sum of $n$ such path-with-self-loop graphs, one for each eigenvalue of~$A$.}
  \label{fig:path_ham}
\end{figure}
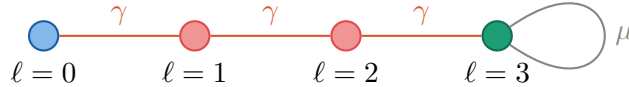

The physical interpretation is transparent (see Figure~\ref{fig:path_ham}): $H^{(\mu)}$ describes a quantum particle on a finite path, with a reflecting boundary at the top ($\ell=0$) and a potential~$\mu$ at the bottom ($\ell=L$) that encodes the effect of the base graph.

\begin{proposition}[Block diagonalization]\label{prop:block}
For each $j\in\{1,\ldots,n\}$, the subspace $\cH_j = \mathrm{span}\{\ket{\alpha_j}\otimes\ket{\ell} : \ell=0,\ldots,L\}$ is invariant under $A_\tower$, and $A_\tower|_{\cH_j}$ is unitarily equivalent to $H^{(\mu_j)}$. Consequently,
\begin{equation}\label{eq:block}
A_\tower \cong \bigoplus_{j=1}^n H^{(\mu_j)}.
\end{equation}
\end{proposition}

\begin{proof}
Let $\ket{\psi}=\ket{\alpha_j}\otimes\ket{\phi}\in\cH_j$. Then:
\begin{align*}
A_\tower\ket{\psi}
&= (A\otimes\proj{L})(\ket{\alpha_j}\otimes\ket{\phi}) + (I\otimes P)(\ket{\alpha_j}\otimes\ket{\phi})\\
&= \mu_j\ket{\alpha_j}\otimes(\braket{L}{\phi}\ket{L}) + \ket{\alpha_j}\otimes(P\ket{\phi})\\
&= \ket{\alpha_j}\otimes H^{(\mu_j)}\ket{\phi}\;\in\;\cH_j.\qedhere
\end{align*}
\end{proof}

When the base graph has repeated eigenvalues $\mu_j=\mu_{j'}$, the path Hamiltonians $H^{(\mu_j)}=H^{(\mu_{j'})}$ are identical, producing eigenvalues with multiplicity equal to that of $\mu_j$ in $\sigma(A)$.

\begin{figure}[ht]
\centering
\includegraphics[width=1.0\textwidth]{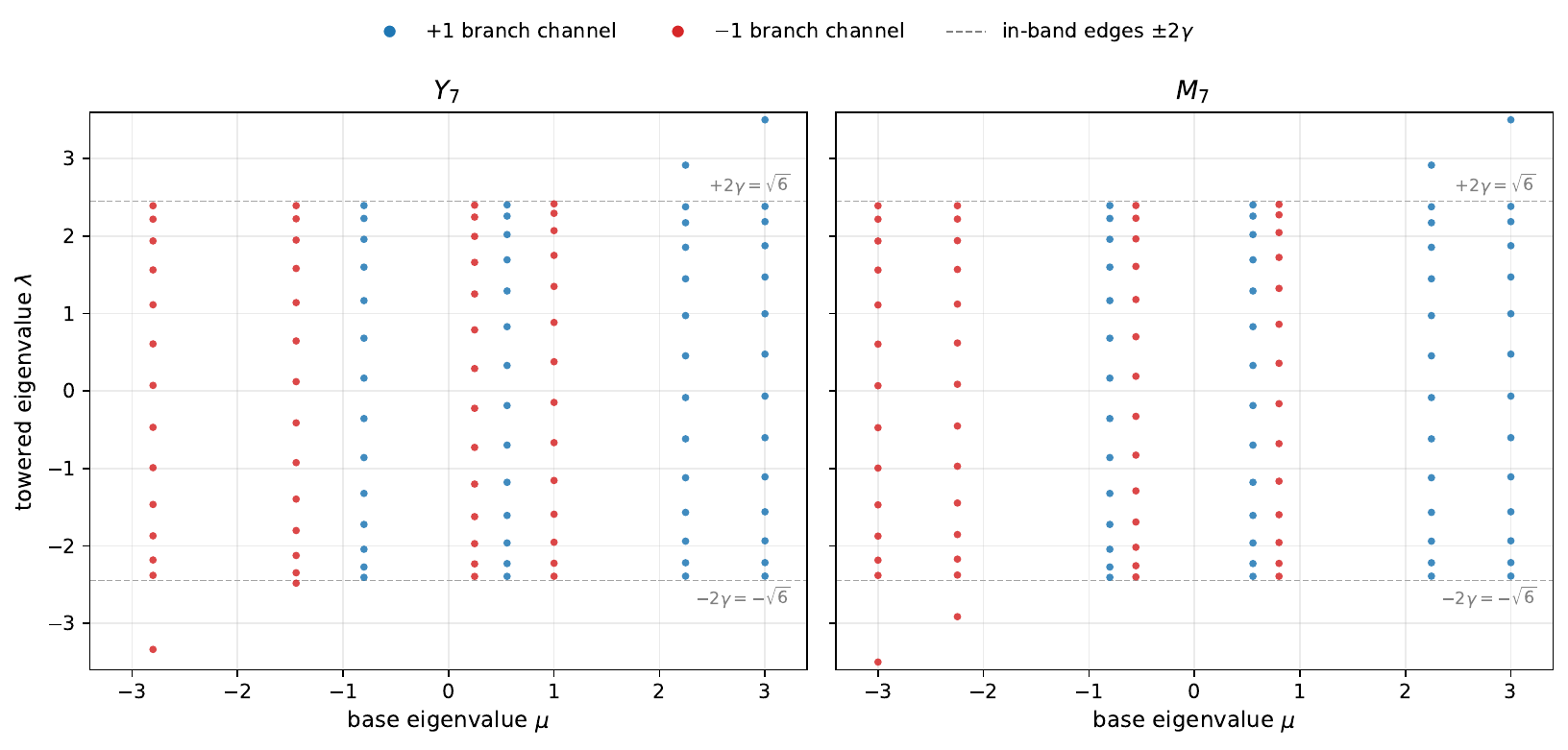}
\caption{Towered spectra of $Y_7$ (left) and $M_7$ (right) at $c=2$, $d=3$, $L = n{-}1 = 13$, visualised through the block decomposition $A_\tower \cong \bigoplus_j H^{(\mu_j)}$ of Proposition~\ref{prop:block}. Each column at base eigenvalue $\mu$ is the spectrum of the path Hamiltonian $H^{(\mu)}$ (Definition~\ref{def:path_ham}), comprising $L+1 = 14$ towered eigenvalues. Colour encodes the branch of $\mu$ in the unified form of Proposition~\ref{prop:unified_evals}: blue for $+1$-branch channels, red for $-1$-branch channels. The two panels share their blue columns exactly (the $+1$ branch coincides between $Y_m$ and $M_m$); the red columns of $M_7$ are shifted by the half-step $\pi/m = \pi/7$ in the cosine argument relative to those of $Y_7$, with the leftmost $M_7$ red column reaching $\mu = -3$ (whose channel produces an out-of-band eigenvalue near $\lambda \approx -3.5$, reflecting the bipartiteness of $M_7$). Dashed lines mark the in-band edges $\pm 2\gamma = \pm\sqrt{6}$; eigenvalues outside this band come from the analytic continuation of the secular equation to the hyperbolic branch (Section~\ref{sec:inband}).}
\label{fig:towered_spectrum}
\end{figure}

\FloatBarrier


\section{The secular equation}\label{sec:secular}

We now determine the eigenvalues of $H^{(\mu)}$ explicitly, using the connection to Chebyshev polynomials of the second kind.

\subsection{Recurrence relations and Chebyshev polynomials}\label{sec:recurrence}

The eigenvalue equation $H^{(\mu)}\ket{\phi}=\lambda\ket{\phi}$ yields three conditions on the amplitudes $\phi_\ell$. At the top we have the \emph{top condition}
\begin{equation}\label{eq:root_bc}
\lambda\phi_0 = \gamma\,\phi_1;
\end{equation}
in the interior we have the \emph{bulk recurrence}
\begin{equation}\label{eq:bulk}
\lambda\phi_\ell = \gamma\,\phi_{\ell-1}+\gamma\,\phi_{\ell+1},\qquad 1\le\ell\le L-1;
\end{equation}
and at the bottom we have the \emph{bottom condition}
\begin{equation}\label{eq:foot_bc}
\lambda\phi_L = \gamma\,\phi_{L-1}+\mu\,\phi_L.
\end{equation}
Setting $x=\lambda/(2\gamma)$, the bulk recurrence becomes $2x\,\phi_\ell = \phi_{\ell-1}+\phi_{\ell+1}$, the defining recurrence for the Chebyshev polynomials of the second kind: $U_0(x)=1$, $U_1(x)=2x$, $U_{\ell+1}(x)=2x\,U_\ell(x)-U_{\ell-1}(x)$.

\begin{lemma}[Path amplitudes]\label{lem:amplitudes}
The sequence $\phi_\ell=U_\ell(x)$ satisfies both the top condition~\eqref{eq:root_bc} and the bulk recurrence~\eqref{eq:bulk}, with $\phi_0=U_0(x)=1$ at the top for all~$\lambda$.
\end{lemma}

\begin{proof}
For the top condition, $\phi_0 = U_0(x) = 1$ and $\phi_1 = U_1(x) = 2x = \lambda/\gamma$, so $\gamma\,\phi_1 = \lambda = \lambda\,\phi_0$, verifying~\eqref{eq:root_bc}. For the bulk recurrence, the Chebyshev defining identity $U_{\ell+1}(x) = 2x\,U_\ell(x) - U_{\ell-1}(x)$, with $x = \lambda/(2\gamma)$, becomes $\gamma\,U_{\ell+1}(x) = \lambda\,U_\ell(x) - \gamma\,U_{\ell-1}(x)$, i.e., $\lambda\,\phi_\ell = \gamma\,\phi_{\ell-1} + \gamma\,\phi_{\ell+1}$, which is~\eqref{eq:bulk}.
\end{proof}

\begin{theorem}[Secular equation]\label{thm:secular}
The sequence $\phi_\ell=U_\ell(x)$ satisfies the bottom condition~\eqref{eq:foot_bc} if and only if $\lambda$ is a root of
\begin{equation}\label{eq:secular}
\boxed{\;\gamma\,U_{L+1}(x) = \mu\,U_L(x),\qquad x=\frac{\lambda}{2\gamma}.\;}
\end{equation}
\end{theorem}

\begin{proof}
By Lemma~\ref{lem:amplitudes}, the sequence $\phi_\ell = U_\ell(x)$ satisfies the top condition~\eqref{eq:root_bc} and the bulk recurrence~\eqref{eq:bulk} for every $\lambda \in \R$, so it suffices to determine for which $\lambda$ this sequence additionally satisfies the bottom condition~\eqref{eq:foot_bc}. Substituting $\phi_\ell = U_\ell(x)$ into~\eqref{eq:foot_bc} gives
\[
\lambda\,U_L(x) \;=\; \gamma\,U_{L-1}(x) + \mu\,U_L(x),
\]
which rearranges to $(\lambda - \mu)\,U_L(x) = \gamma\,U_{L-1}(x)$. The Chebyshev defining recurrence at index~$L$, namely $U_{L+1}(x) = 2x\,U_L(x) - U_{L-1}(x)$, yields $U_{L-1}(x) = 2x\,U_L(x) - U_{L+1}(x)$. Substituting and using $2\gamma x = \lambda$,
\[
(\lambda - \mu)\,U_L(x) \;=\; \gamma\bigl(2x\,U_L(x) - U_{L+1}(x)\bigr) \;=\; \lambda\,U_L(x) - \gamma\,U_{L+1}(x).
\]
Cancelling $\lambda\,U_L(x)$ from both sides yields $\gamma\,U_{L+1}(x) = \mu\,U_L(x)$, which is the secular equation~\eqref{eq:secular}. Each step is reversible, so the bottom condition holds if and only if $\lambda$ is a root of the secular equation.
\end{proof}

\begin{proposition}[Root count and simplicity]\label{prop:roots}
For any $\mu\in\R$, the secular equation has exactly $L{+}1$ distinct real roots. This follows from the classical result~\cite{golub2013matrix} that a symmetric tridiagonal matrix with all nonzero off-diagonal entries has simple eigenvalues.
\end{proposition}

\subsection{In-band and out-of-band eigenvalues}\label{sec:inband}

The spectral band of the infinite weighted path is $[- 2\gamma,2\gamma]$, corresponding to $|x|<1$. For in-band eigenvalues, the trigonometric parametrization $x=\cos\theta$ with $\theta\in(0,\pi)$ gives $U_\ell(\cos\theta) = \sin((\ell{+}1)\theta)/\sin\theta$, the standing waves along the path. The secular equation becomes
\begin{equation}\label{eq:secular_trig}
\gamma\,\sin\bigl((L{+}2)\theta\bigr) = \mu\,\sin\bigl((L{+}1)\theta\bigr),
\end{equation}
which is crucial for the efficient eigenvalue computation described in Section~\ref{sec:efficient}.

Out-of-band eigenvalues ($|x|>1$) correspond to modes exponentially localized near the bottom, with top weights that decay exponentially in~$L$. These contribute negligibly to the quantum walk dynamics and can be safely excluded from the computation.


\section{Top weights and the return amplitude}\label{sec:root_weight}

Having determined the eigenvalues of $H^{(\mu)}$ via the secular equation, we now compute their spectral weights at the top of each tower and assemble the return amplitude that drives the quantum walk.

\subsection{The top weight formula}

The spectral weight of an eigenvalue at the top is the squared overlap of the corresponding normalised eigenstate with the top basis vector $\ket{v,0}$.

\begin{proposition}[Top weight]\label{prop:root_weight}
The spectral weight of eigenvalue $\lambda_{i,k}$ at the top $(v,0)$ is
\begin{equation}\label{eq:root_weight}
\boxed{\;W_{i,k}(v) = \frac{p_i(v)}{\displaystyle\sum_{\ell=0}^L U_\ell(x_{i,k})^2},\;}
\end{equation}
where $x_{i,k}=\lambda_{i,k}/(2\gamma)$ and $p_i(v)=\sum_{s=1}^{m_i}|\braket{v}{\alpha_{i,s}}|^2$. For vertex-transitive $G$: $p_i(v)=m_i/n$ for all~$v$.
\end{proposition}

\begin{proof}
The normalized eigenstate is $\ket{\hat\Psi_{i,k,s}}=\ket{\alpha_{i,s}}\otimes\ket{\phi_{i,k}}/\|\phi_{i,k}\|$. Since $\phi_{i,k,0}=U_0(x_{i,k})=1$, the squared overlap with the top $\ket{v,0}$ is $|\braket{v}{\alpha_{i,s}}|^2/\|\phi_{i,k}\|^2$. Summing over degenerate states gives~\eqref{eq:root_weight}.
\end{proof}

The denominator admits a closed form that is essential for efficient computation:
\begin{equation}\label{eq:norm_closed}
\sum_{\ell=0}^L U_\ell(\cos\theta)^2 = \frac{(2L{+}3)\sin\theta - \sin\bigl((2L{+}3)\theta\bigr)}{4\sin^3\theta}.
\end{equation}
This follows from the identity $U_\ell(\cos\theta)=\sin((\ell{+}1)\theta)/\sin\theta$ and the summation formula $\sum_{k=1}^M \sin^2(k\theta)=M/2 - \sin((2M{+}1)\theta)/(4\sin\theta) + 1/4$. The closed form evaluates in $O(1)$ per eigenvalue, avoiding the $O(L)$ cost of the Chebyshev recurrence.

\subsection{The return amplitude}

Summing the weighted eigenvalue contributions over all $n(L{+}1)=n^2$ eigenpairs gives the return amplitude as a sum of complex exponentials.

\begin{proposition}[Return amplitude]\label{prop:return}
For a vertex-transitive base graph, the return amplitude at the top is
\begin{equation}\label{eq:signal}
f_G(t) = \bra{v,0}\exp(-iA_\tower t)\ket{v,0}
= \sum_{i=1}^r\sum_{k=1}^{L+1} W_{i,k}\,e^{-i\lambda_{i,k}t},
\end{equation}
with $W_{i,k}=m_i/(n\|\phi_{i,k}\|^2)$. All weights satisfy $W_{i,k}\ge 0$ and $\sum_{i,k}W_{i,k}=1$.
\end{proposition}


\section{The quantum algorithm}\label{sec:algorithm}

\subsection{Problem formulation}\label{sec:problem_formulation}

\noindent\textbf{Input:} An apex label $x_0=\iota((u,\epsilon))$ and black-box oracle $\cO_G$ for the spired graph~$G_\spire$.

\noindent\textbf{Promise:} The hidden base graph $G$ is one of $r$ candidates $G_1,\ldots,G_r$, all $d$-regular on $n$ vertices.

\noindent\textbf{Goal:} Identify $G$ with probability $\ge 1-\delta$.

We assume perfect Hamiltonian simulation: given black-box access via the oracle $\cO_G$, the quantum computer can implement $\exp(-iH_\spire t)$ exactly for any time~$t$.  This is a standard idealization: decisions based on a close approximation $\widetilde U\approx U$ remain close to those based on $U$ by standard operator-norm arguments, and we suppress this analysis to focus on the structural features of the problem.

\subsection{The cross-graph spectral test}\label{sec:cross_graph}

Given $r$ candidate base graphs $G_1, \ldots, G_r$, the algorithm estimates the return amplitude $f(t) = \bra{x_0}\exp(-iH_\spire t)\ket{x_0}$ at a classically precomputed optimal time~$t^*$.  The classical precomputation evaluates $f_{G_1}, \ldots, f_{G_r}$ on a fine grid using the spectral pipeline of Section~\ref{sec:efficient}, and chooses
\begin{equation}\label{eq:tstar}
t^* \;=\; \arg\max_{t\in[0,T_{\max}]}\, \min_{j\neq k}|f_{G_j}(t)-f_{G_k}(t)|.
\end{equation}

\paragraph{Quantum estimation via the Hadamard test.}  Let $U = \exp(-iH_\spire t^*)$.  A standard Hadamard test (Figure~\ref{fig:hadamard}) suffices: by Theorem~\ref{thm:restriction}, the algorithm's dynamics on the Krylov subspace of $\ket{x_0}$ coincides with evolution under the towered-graph adjacency matrix $A_\tower$, so the single controlled-$U$ at the precomputed time~$t^*$ already encodes the spectral signature of the hidden base graph.  Repeating $N_\mathrm{rep}$ times in each branch yields an estimator $\widetilde f\in\mathbb{C}$ with $\widetilde f\approx f(t^*)$.  Structurally, this echoes the quantum-feature paradigm of Havl\'i\v{c}ek et~al.~\cite{havlicek19supervised}: the quantum subroutine estimates a single complex number that would be classically expensive to compute, and the subsequent decision is purely classical.

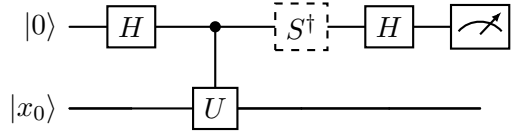
\begin{figure}[ht]
\centering
\begin{quantikz}
\lstick{$\ket{0}$}   & \gate{H} & \ctrl{1} & \gate[style={dashed}]{S^\dagger} & \gate{H} & \meter{} \\
\lstick{$\ket{x_0}$} & \qw      & \gate{U} & \qw                              & \qw      & \qw
\end{quantikz}
\caption{Hadamard-test circuit for $f(t^*)=\bra{x_0}U\ket{x_0}$ with $U=\exp(-iH_\spire t^*)$, implemented via the oracle $\cO_G$.  System register $\ket{x_0}=\ket{\iota(a_u)}$ is the apex label; each shot yields $\pm 1$ with expectation $\mathrm{Re}\,f(t^*)$, or $\mathrm{Im}\,f(t^*)$ if $S^\dagger$ (dashed) is applied.}
\label{fig:hadamard}
\end{figure}

\paragraph{Decision rule.}  Output the candidate whose prediction is closest to the measurement in $\mathbb{C}$:
\[
\hat\jmath \;=\; \arg\min_{j\in\{1,\ldots,r\}} \bigl|f_{G_j}(t^*) - \widetilde f\bigr|.
\]

\begin{remark}[Pair-tournament alternative]\label{rem:pair_tournament}
For $r > 2$, an alternative to the single global $t^*$ is a pairwise tournament that runs a Hadamard test at each pair's optimal $t^*_{j,k} = \arg\max_t |f_{G_j}(t)-f_{G_k}(t)|$ and outputs the candidate that wins all $r-1$ of its contests.  This trades $\binom{r}{2}$ Hadamard-test instances for the freedom to use each pair's own optimal time, and can be more efficient when no single $t^*$ distinguishes all pairs simultaneously well.  We do not pursue this trade-off here.
\end{remark}

\subsection{Measurement budget}\label{sec:budget}

\begin{definition}[Peak distinguishability]\label{def:peak}
The peak distinguishability of $r$ candidate graphs over the time horizon $[0,T_{\max}]$ is
$\Dis(m) = \max_{t\in[0,T_{\max}]}\min_{j\neq k}|f_{G_j}(t)-f_{G_k}(t)|$.
\end{definition}

\begin{theorem}[Measurement budget]\label{thm:budget}
For the cross-graph spectral test of Section~\ref{sec:cross_graph}, success probability $\ge 1-\delta$ is achieved by repeating the Hadamard test
\begin{equation}\label{eq:budget_single}
N_{\mathrm{rep}} \;=\; \left\lceil\frac{2z^2}{\Dis(m)^2}\right\rceil
\end{equation}
times in each of the real and imaginary branches, where $z=\Phi^{-1}(1-\delta/(r-1))$.
\end{theorem}

\begin{proof}
The Hadamard test of Section~\ref{sec:cross_graph} returns a complex estimator $\widetilde f$ whose real and imaginary parts are independent unbiased estimators of $\mathrm{Re}\,f(t^*)$ and $\mathrm{Im}\,f(t^*)$, each with Gaussian noise of standard deviation at most $1/\sqrt{N_\mathrm{rep}}$.  Under the true hypothesis $j^*\in\{1,\ldots,r\}$, the true value lies at $f_{G_{j^*}}(t^*)$, separated from each other candidate prediction $f_{G_k}(t^*)$ ($k\neq j^*$) by at least $\Dis(m)$ in $\mathbb{C}$ (by definition of $t^*$).  The decision rule misclassifies only if $\widetilde f$ falls closer to some wrong prediction $f_{G_k}(t^*)$ than to $f_{G_{j^*}}(t^*)$, which for that $k$ requires a noise displacement of magnitude $\ge\Dis(m)/2$.  By the Gaussian tail bound, each such pairwise misclassification event has probability $\le 2\Phi(-\Dis(m)\sqrt{N_\mathrm{rep}}/2)$; a union bound over the $r-1$ wrong candidates gives a total of $\le 2(r-1)\Phi(-\Dis(m)\sqrt{N_\mathrm{rep}}/2)$, and setting this below $\delta$ gives~\eqref{eq:budget_single}.
\end{proof}

\subsection{The Parseval sum and choice of time horizon}\label{sec:horizon_choice}

The difference signal $\Delta f(t)=f_{G_A}(t)-f_{G_B}(t)=\sum_\lambda c_\lambda e^{-i\lambda t}$ is a sum of $\Theta(n^2)$ complex exponentials, where $c_\lambda = W_\lambda^{(G_A)} - W_\lambda^{(G_B)}$ is the difference of top weights at eigenvalue~$\lambda$.

\begin{definition}[Parseval sum]\label{def:parseval}
The \emph{Parseval sum} of two base graphs $G_A, G_B$ is the time-averaged squared distinguishability:
\begin{equation}\label{eq:parseval_def}
\Par^\Delta \;=\; \Par^\Delta_{A,B} \;=\; \lim_{T\to\infty}\frac{1}{T}\int_0^T |\Delta f(t)|^2\,dt.
\end{equation}
\end{definition}

\begin{proposition}[Parseval identity]\label{prop:parseval}
\begin{equation}\label{eq:parseval_sum}
\Par^\Delta \;=\; \sum_\lambda |c_\lambda|^2,
\end{equation}
where $c_\lambda=W_\lambda^{(G_A)}-W_\lambda^{(G_B)}$ is the weight difference at eigenvalue~$\lambda$ (with $c_\lambda=W_\lambda^{(G_A)}$ if $\lambda\notin\sigma(H_{G_B})$, and similarly with the sign reversed).
\end{proposition}

\begin{proof}
Expand $|\Delta f(t)|^2 = \sum_{\lambda,\lambda'} c_\lambda\bar c_{\lambda'}e^{-i(\lambda-\lambda')t}$.  The time average over $[0,T]$ gives $|c_\lambda|^2$ for the diagonal terms ($\lambda=\lambda'$) and $O(1/(T|\lambda-\lambda'|))$ for off-diagonal terms, which vanish as $T\to\infty$.
\end{proof}

\begin{proposition}[Peak vs.\ Parseval bound]\label{prop:peak_vs_parseval}
For any time horizon $T_{\max}$ large enough that the Cesàro average is close to its limit,
\begin{equation}\label{eq:peak_parseval}
\Dis(m)^2 \;=\; \max_{t\in[0,T_{\max}]}|\Delta f(t)|^2 \;\ge\; \Par^\Delta\,(1-o(1)).
\end{equation}
\end{proposition}

\begin{proof}
The time average of $|\Delta f(t)|^2$ converges to~$\Par^\Delta$, and the time average cannot exceed the maximum.
\end{proof}

\begin{remark}[Multi-candidate generalization]\label{rem:multi_candidate_parseval}
For $r > 2$ candidates, the Parseval analysis above applies to each pair $(G_j, G_k)$ separately, with the per-pair Parseval $\Par^\Delta_{j,k}$ providing the natural per-pair scale; sharp bounds on the global max-min $\Dis(m)$ of Definition~\ref{def:peak} are left for future work.
\end{remark}

The ratio $\Dis(m)^2/\Par^\Delta$ quantifies the \emph{constructive-interference advantage}: by how much the single-time algorithm (measuring at the optimal time~$t^*$) outperforms a strategy based on the time-averaged signal.

Resolving the constructive interference requires a time horizon long enough to separate the relevant frequency differences.  The eigenvalue spacing within a single channel (fixed $\mu$) is $\Theta(1/n)$, while the finest inter-channel spacings can be $\Theta(1/n^2)$, suggesting the \emph{quadratic} time horizon
\begin{equation}\label{eq:quadratic_horizon}
T_{\max} = c\cdot m^2
\end{equation}
with a moderate constant $c\ge 1$. Numerical experiments confirm $c=1$ is both necessary and sufficient: $\Dis(m)$ drops by roughly a factor of four at $c=0.5$ (e.g., from $0.249$ to $0.068$ at $m=8$, in line with the $\Omega(n^2)$ Fourier-uncertainty bound) and saturates to numerical precision for any $c\ge 1$. We adopt $T_{\max}=m^2$ throughout.

\begin{remark}[Classical precomputation cost]
The classical precomputation of $t^*$ requires evaluating the return amplitude at $O(m^2)$ time points, at cost $O(n^2)$ per evaluation (or $O(n^2\log n)$ total using the nonuniform fast Fourier transform (NUFFT)), for a total of $O(n^4)$ operations with the brute-force method or $O(n^2\log^2 n)$ with the NUFFT. Since this is a one-time computation performed before any quantum queries, it does not affect the quantum complexity. The search over $t$ is also trivially parallelizable across time intervals.
\end{remark}


\section{Efficient classical computation}\label{sec:efficient}

The spectral decomposition of Sections~\ref{sec:spectral}--\ref{sec:secular} transforms the problem of evaluating the return amplitude from an exponentially large matrix exponential into a structured computation that can be performed in polynomial time. We describe two complementary implementations: the \emph{direct method}, which constructs and diagonalizes the full towered graph Hamiltonian without invoking any spectral theory, and the \emph{SERF method} (secular equation root-finding), which exploits the full theoretical machinery to scale to large~$m$. Agreement between the two at small~$m$ validates the entire theoretical chain.

\subsection{The computational pipeline}\label{sec:pipeline}

The evaluation of $f_G(t)$ at $T$ time points proceeds in two stages:

\textbf{Stage~1: Spectral data.} For each of the $n$ base-graph eigenvalues $\mu_j$, find the $n$ eigenvalues $\lambda_{j,k}$ of the path Hamiltonian $H^{(\mu_j)}$ and compute the top weights $W_{j,k}$. The total spectral data consists of $n^2$ eigenvalue--weight pairs $\{(\lambda_{j,k}, W_{j,k})\}$.

\textbf{Stage~2: Signal evaluation.} Evaluate $f(t) = \sum_{j,k} W_{j,k}\,e^{-i\lambda_{j,k}t}$ at all $T$ time points.

We have implemented two principal methods, described below.

\subsection{The direct method: full towered graph diagonalization}\label{sec:direct}

The direct method is the \emph{ground-truth} implementation that performs no spectral analysis: it just constructs the $n^2\times n^2$ towered-graph adjacency matrix $A_\tower = A\otimes\proj{L} + I_n\otimes P$ (the tensor-product form of Proposition~\ref{prop:tensor}) and diagonalizes it with a single call to \texttt{numpy.linalg.eigh}. The top weight for eigenstate $\ket{\psi_k}$ is simply
\[
W_k = |\braket{v{=}0,\,\ell{=}0}{\psi_k}|^2,
\]
the squared magnitude of the first component of the $k$-th eigenvector. No channel decomposition, no Chebyshev formulas, no secular equation; just one matrix diagonalization. This makes the direct method the most trustworthy validation tool: any error in the spectral theory of Sections~\ref{sec:spectral}--\ref{sec:secular} would show up as a discrepancy between the direct and SERF results.

The limitation is computational: the matrix has dimension $n(L{+}1) = n^2$, so the cost is $O(n^6)$ and memory is $O(n^4)$. In practice, the direct method is feasible for $m\le 32$ (matrix dimension $4096$, $\sim 30$\,s) and borderline at $m=64$ (dimension $16384$).

\subsection{The SERF method: secular equation root-finding}\label{sec:serf}

The production method exploits the full theoretical machinery. For Stage~1, it bypasses matrix diagonalization entirely by solving the secular equation in $\theta$-space. Setting $x=\cos\theta$ with $\theta\in(0,\pi)$, the secular equation~\eqref{eq:secular} becomes $G(\theta) = \gamma\sin((L{+}2)\theta) - \mu\sin((L{+}1)\theta) = 0$.

The $L{+}1$ in-band roots are bracketed by Cauchy interlacing: the $(L{+}1)\times(L{+}1)$ path Hamiltonian $H^{(\mu)}$ has a leading $L\times L$ principal submatrix whose eigenvalues are $2\gamma\cos(k\pi/(L{+}1))$ for $k=1,\ldots,L$. By the interlacing theorem, each interval $[k\pi/(L{+}1),\,(k{+}1)\pi/(L{+}1)]$ contains at most one eigenvalue. We set up an $(n)\times(L{+}1)$ array of bracket pairs and perform \emph{vectorized bisection}: all $n(L{+}1)$ brackets are bisected simultaneously, with 53 bisection steps achieving full double precision. Top weights are computed via the closed-form norm sum~\eqref{eq:norm_closed} in $\theta$-space. The total cost of Stage~1 is $O(n^2)$ with a small constant.

For Stage~2, we observe that $f(t) = \sum_k W_k\,e^{-i\lambda_k t}$ has the form of a \emph{type-3 nonuniform discrete Fourier transform}. We use the FINUFFT library~\cite{barnett19finufft,finufft_software} (Flatiron Institute Nonuniform Fast Fourier Transform), which computes such sums in $O((N+T)\log(N+T)\cdot\log(1/\varepsilon))$ operations, reducing Stage~2 from $O(n^2 T)$ to $O(n^2\log n)$.

\begin{remark}[Memory-efficient chunking for large $m$]
For $m=4096$ ($n=8192$), the spectral data consists of $\sim 134$ million active eigenvalue--weight pairs per graph pair, and the coarse time grid has $\sim 33.5$ million points. The SERF method addresses this by \emph{chunking} both the time points (during the coarse scan, $\sim 4\times 10^6$ per block) and the eigenvalues (during refinement, $\sim 30\times 10^6$ per block). Since $f(t) = \sum_k W_k e^{-i\lambda_k t}$ is a linear sum, the eigenvalue blocks can be computed independently and accumulated. This reduces peak memory from $\sim 6$~GB to $\sim 1.5$~GB per NUFFT call, enabling execution on free-tier cloud platforms.
\end{remark}

\begin{remark}[Why FINUFFT]
We chose FINUFFT because it is (i) optimized for double precision with guaranteed error bounds, (ii) available as a compiled C++ library with a Python interface (\texttt{pip install finufft}), and (iii) widely used in computational physics and signal processing. The NUFFT acceleration is critical for large~$m$: at $m=2048$, the brute-force signal evaluation would require $\sim 10^{14}$ operations; the NUFFT reduces this to $\sim 10^8$.
\end{remark}

\subsection{Peak finding with refinement}\label{sec:peak_finding}

To locate $\Dis(m) = \max_t|\Delta f(t)|$ accurately despite the large time horizon $T_{\max}=m^2$, we employ a two-stage strategy:

\textbf{Coarse scan.} Evaluate $|\Delta f(t)|$ on a uniform grid with spacing $\delta t = 0.5$ over $[0,T_{\max}]$, yielding $O(m^2)$ grid points. The NUFFT makes this affordable even for $m=4096$ ($\sim 33.5\times 10^6$ grid points). We identify the top~$3$--$5$ candidate peaks.

\textbf{Refinement.} Around each candidate peak, we evaluate $|\Delta f(t)|$ on a fine grid with spacing $\delta t = 0.01$ over a window of width $\pm 10$, then $\pm 5$. This resolves the peak to 4--5 significant digits beyond the coarse-scan estimate. We verified by convergence testing (varying $\delta t$ from 0.5 down to 0.001) that $\delta t = 0.01$ is more than sufficient: the relative change in $\Dis(m)$ between $\delta t=0.01$ and $\delta t=0.001$ is below $10^{-6}$ for all~$m$.

\subsection{Cross-validation}\label{sec:cross_validation}

The direct method (full towered graph, no theory) and the SERF method (root-finding + NUFFT, full theory) were run independently for $m\in\{4,8,16,32\}$. The optimal time~$t^*$ and the Parseval sum~$\Par^\Delta$ agree to full precision at all four sizes (Table~\ref{tab:crossval}). The values of $\Dis(m)$ agree to machine precision at $m=32$ and show small discrepancies at $m\le 16$ due to \emph{out-of-band eigenvalues}: when a base-graph eigenvalue $|\mu_j|> 2\gamma$, the towered Hamiltonian has eigenvalues outside the spectral band $(- 2\gamma,2\gamma)$. The direct method captures these via full diagonalization, while the SERF method's secular equation finds only in-band roots. The out-of-band top weights decay as $e^{-2\kappa L}$ where $\kappa = \mathrm{acosh}(|\mu_j|/(2\gamma))$ and $L=n-1$, so they vanish exponentially with~$m$.

\begin{table}[ht]
\centering
\begin{tabular}{rlll}
\toprule
$m$ & Direct ($\Dis(m)$) & SERF ($\Dis(m)$) & $|\Dis_{\mathrm{dir}} - \Dis_{\mathrm{SERF}}|$\\
\midrule
4 & 0.345224630 & 0.345226626 & $2.0\times 10^{-6}$\\
8 & 0.249073792 & 0.249073613 & $1.8\times 10^{-7}$\\
16 & 0.125623015 & 0.125623267 & $2.5\times 10^{-7}$\\
32 & 0.069249693 & 0.069249693 & $4.0\times 10^{-12}$\\
\bottomrule
\end{tabular}
\caption{Cross-validation of $\Dis(m)$ between the direct method (full towered graph diagonalization) and the SERF method (secular equation root-finding + NUFFT). The discrepancy at $m\le 16$ is due to out-of-band eigenvalues whose top weights decay exponentially with~$n$; by $m=32$ the agreement is at machine precision.}
\label{tab:crossval}
\end{table}

All computations in this paper, including the largest instance $m=5121$ ($n=10242$, ${\sim}210$ million eigenvalue--weight pairs), complete in under 20 minutes on a consumer desktop with a 6-core Intel Core i5 CPU and 32~GB of RAM.


\section{Numerical evidence}\label{sec:numerics}

All computations use $L=n{-}1$ and $\gamma=\sqrt{6}/2$ (corresponding to $d=3$, $D=6$).  Results for $m\le 32$ are cross-validated between the direct method and the SERF method; all results through $m=5121$ are computed by the SERF method.

\subsection{Setup}\label{sec:numerics_setup}

The cross-graph spectral test of Section~\ref{sec:cross_graph} requires precomputing the peak distinguishability
\[
\Dis(m) \;=\; \max_{t\in[0,m^2]}\bigl|f_{Y_m}(t)-f_{M_m}(t)\bigr|
\]
on the quadratic time horizon $T_{\max}=m^2$.  By the Parseval identity~\eqref{eq:parseval_sum}, $\Par^\Delta=\sum_\lambda|c_\lambda|^2$ with $c_\lambda=W_\lambda^{(Y)}-W_\lambda^{(M)}$, and Proposition~\ref{prop:peak_vs_parseval} gives $\Dis(m)^2\ge \Par^\Delta$.  The two quantities $\Dis(m)$ and $\Par^\Delta$ together control the measurement budget through Theorem~\ref{thm:budget}.

\subsection{Scaling}\label{sec:scaling_observations}

We compute the cross-graph peak $\Dis(m)$ and the Parseval sum $\Par^\Delta$ on the quadratic time horizon $T_{\max}=m^2$ at $80$ values of $m$ spanning $m=4$ to $m=5121$: the $11$~doubling pairs $(2^k, 2^k{+}1)$, $k=2,\ldots,12$, supplemented by $9$~geometric mid-pairs $(m, m{+}1)$ at $m\approx \tfrac{3}{2}\cdot 2^k$, $k=2,\ldots,10$, $16$~quarter-octave pairs $(m, m{+}1)$ at $m\approx \tfrac{5}{4}\cdot 2^k$ and $m\approx \tfrac{7}{4}\cdot 2^k$, $k=3,\ldots,10$, three further quarter-octave pairs $(2560, 2561)$ at $m\approx \tfrac{5}{4}\cdot 2^{11}$, $(3584, 3585)$ at $m\approx \tfrac{7}{4}\cdot 2^{11}$, and $(5120, 5121)$ at $m\approx \tfrac{5}{4}\cdot 2^{12}$, and one half-quarter-octave check-point pair $(4608, 4609)$ at $m\approx \tfrac{9}{8}\cdot 2^{12}$.  Table~\ref{tab:Dm} lists the doubling-pair subset; the figure and the fits in this section use all $80$ points.

\begin{table}[ht]
\centering
\begin{tabular}{cccccccr}
\toprule
$m$ & $n$ & $\Dis(m)$ & $\Dis(m)^2$ & $\Par^\Delta$ & $n^2 \Par^\Delta$ & $\Dis(m)^2/\Par^\Delta$ & $N_{\mathrm{rep}}$\\
\midrule
4 & 8 & 0.345227 & $1.19\times 10^{-1}$ & $3.85\times 10^{-2}$ & 2.46 & 3.10 & 46\\
5 & 10 & 0.280534 & $7.87\times 10^{-2}$ & $2.56\times 10^{-2}$ & 2.56 & 3.07 & 69\\
8 & 16 & 0.249074 & $6.20\times 10^{-2}$ & $1.06\times 10^{-2}$ & 2.72 & 5.83 & 88\\
9 & 18 & 0.171994 & $2.96\times 10^{-2}$ & $8.50\times 10^{-3}$ & 2.75 & 3.48 & 183\\
16 & 32 & 0.125623 & $1.58\times 10^{-2}$ & $2.79\times 10^{-3}$ & 2.86 & 5.65 & 343\\
17 & 34 & 0.140713 & $1.98\times 10^{-2}$ & $2.49\times 10^{-3}$ & 2.87 & 7.96 & 274\\
32 & 64 & 0.069250 & $4.80\times 10^{-3}$ & $7.15\times 10^{-4}$ & 2.93 & 6.71 & 1\,129\\
33 & 66 & 0.065811 & $4.33\times 10^{-3}$ & $6.73\times 10^{-4}$ & 2.93 & 6.44 & 1\,250\\
64 & 128 & 0.040552 & $1.64\times 10^{-3}$ & $1.81\times 10^{-4}$ & 2.96 & 9.09 & 3\,292\\
65 & 130 & 0.038529 & $1.48\times 10^{-3}$ & $1.75\times 10^{-4}$ & 2.97 & 8.46 & 3\,646\\
128 & 256 & 0.019818 & $3.93\times 10^{-4}$ & $4.55\times 10^{-5}$ & 2.98 & 8.63 & 13\,781\\
129 & 258 & 0.020014 & $4.01\times 10^{-4}$ & $4.48\times 10^{-5}$ & 2.98 & 8.94 & 13\,511\\
256 & 512 & 0.011070 & $1.23\times 10^{-4}$ & $1.14\times 10^{-5}$ & 2.99 & 10.74 & 44\,164\\
257 & 514 & 0.011506 & $1.32\times 10^{-4}$ & $1.13\times 10^{-5}$ & 2.99 & 11.69 & 40\,878\\
512 & 1024 & 0.005841 & $3.41\times 10^{-5}$ & $2.86\times 10^{-6}$ & 3.00 & 11.94 & 158\,631\\
513 & 1026 & 0.005854 & $3.43\times 10^{-5}$ & $2.85\times 10^{-6}$ & 3.00 & 12.04 & 157\,938\\
1024 & 2048 & 0.003030 & $9.18\times 10^{-6}$ & $7.15\times 10^{-7}$ & 3.00 & 12.85 & 589\,340\\
1025 & 2050 & 0.002933 & $8.60\times 10^{-6}$ & $7.13\times 10^{-7}$ & 3.00 & 12.06 & 629\,125\\
2048 & 4096 & 0.001623 & $2.63\times 10^{-6}$ & $1.79\times 10^{-7}$ & 3.00 & 14.74 & 2\,054\,672\\
2049 & 4098 & 0.001710 & $2.92\times 10^{-6}$ & $1.79\times 10^{-7}$ & 3.00 & 16.38 & 1\,850\,297\\
4096 & 8192 & 0.000891 & $7.93\times 10^{-7}$ & $4.47\times 10^{-8}$ & 3.00 & 17.75 & 6\,822\,742\\
4097 & 8194 & 0.000891 & $7.93\times 10^{-7}$ & $4.47\times 10^{-8}$ & 3.00 & 17.75 & 6\,824\,343\\
\bottomrule
\end{tabular}
\caption{Cross-graph scaling: $\Dis(m)$ and $\Par^\Delta$ for prism vs.\ M\"obius ladder at the $11$~doubling pairs $(2^k, 2^k{+}1)$.  Each pair samples both bipartite phases at the same scale (Corollary~\ref{cor:exactly_one_bipartite}); the resulting parity oscillation in $\Dis(m)^2/\Par^\Delta$ shrinks with $m$ (visible in the rightmost two columns).  The dimensionless quantity $n^2 \Par^\Delta$ converges to $3$, and the constructive-interference factor $\Dis(m)^2/\Par^\Delta$ grows polylogarithmically in~$m$.  An additional $9$~geometric mid-pairs, $18$~quarter-octave pairs, and one half-quarter-octave pair (Section~\ref{sec:scaling_observations}, $56$ further points) are used in Figure~\ref{fig:scaling} and the fits but omitted from the table for compactness.}
\label{tab:Dm}
\end{table}

\begin{figure}[ht]
    \centering
    \includegraphics[width=1.0\textwidth]{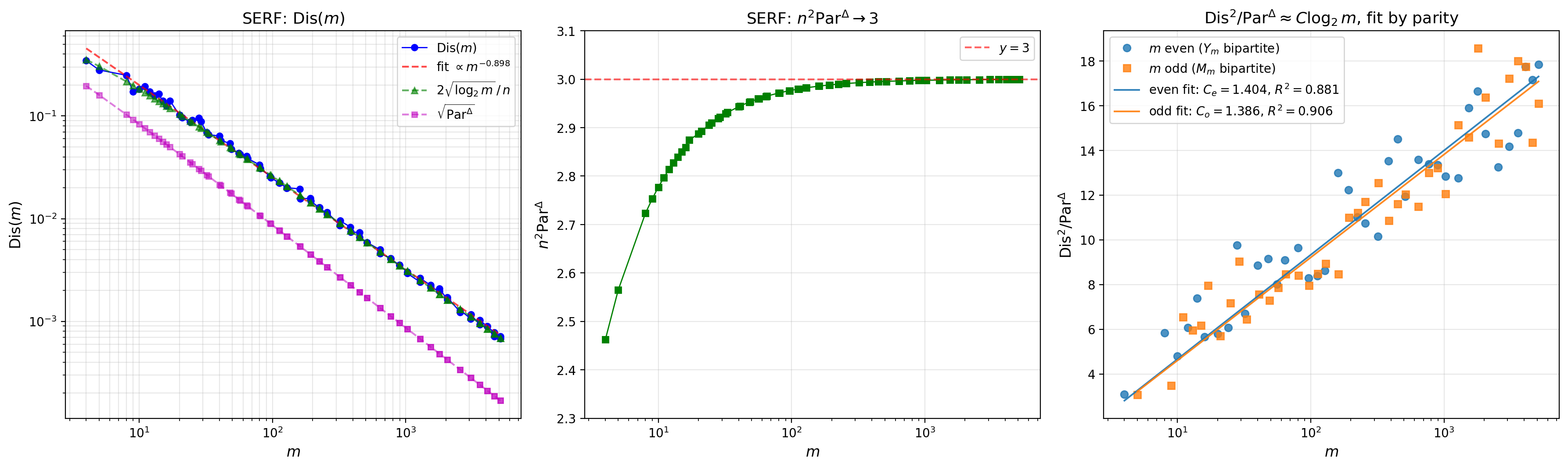}
    \caption{Cross-graph scaling at $80$ values of $m$ spanning $m=4$ to $5121$ in $40$ pairs $(m, m{+}1)$ ($11$ doublings $(2^k,2^k{+}1)$ plus $9$ geometric mid-pairs at $m\approx 1.5\cdot 2^k$ plus $19$ quarter-octave pairs at $m\approx 1.25\cdot 2^k$ and $m\approx 1.75\cdot 2^k$ plus a half-quarter-octave pair at $m\approx 1.125\cdot 2^{12}$). \emph{Left:} log-log plot of $\Dis(m)$ (blue), the prediction $2\sqrt{\log_2 m}/n$ (green triangles), and the Parseval bound $\sqrt{\Par^\Delta}$ (magenta). \emph{Center:} convergence of $n^2 \Par^\Delta\to 3$. \emph{Right:} $\Dis(m)^2/\Par^\Delta$ separated by parity (blue circles, even~$m$, $Y_m$ bipartite; orange squares, odd~$m$, $M_m$ bipartite), with independent linear-in-$\log_2 m$ fits on each parity class (solid lines, fit constants in the legend).}
    \label{fig:scaling}
\end{figure}

\FloatBarrier

Three patterns emerge from Table~\ref{tab:Dm} and Figure~\ref{fig:scaling}.

\paragraph{Parseval sum converges to a clean rational limit.} The dimensionless quantity $n^2 \Par^\Delta$ rises monotonically from $2.46$ at $m=4$ to $3.000$ at $m=4096$, with the deficit $3 - n^2 \Par^\Delta$ halving at each doubling of~$m$.  Equivalently, $\Par^\Delta = (3 - O(1/n))/n^2$, which gives the Parseval-style lower bound
\[
\Dis(m) \;\ge\; \sqrt{\Par^\Delta} \;=\; \Omega(1/n).
\]

\paragraph{The peak follows a closed-form approximation.} $\Dis(m)$ matches $2\sqrt{\log_2 m}/n$ to within $\approx 15\%$ across $m\in[64,5121]$, three orders of magnitude with no systematic drift; the ratio of observed to predicted stays in $[0.93,1.15]$ on this range, with $90\%$ of points within $[0.95,1.11]$.  At $m=4096$, $\Dis=8.906\times 10^{-4}$, agreeing with the prediction to $\approx 5\%$.  The asymptotic identity is therefore
\begin{equation}\label{eq:D_scaling}
\Dis(m) \;\approx\; \frac{2\sqrt{\log_2 m}}{n},
\end{equation}
with a parity-dependent oscillation of up to $\pm 15\%$ at intermediate $m$.  Two single-exponent power-law fits $\Dis(m)\propto m^\alpha$ illustrate this: on the direct-method range $m\le 33$ we obtain $\alpha\approx -0.76$, while on the upper SERF range $m\ge 64$ we obtain $\alpha\approx -0.91$; both are consistent with the closed-form effective exponent $\alpha_\mathrm{eff}(m)=-1+1/(2\ln m)$ evaluated at the geometric mean of the corresponding sub-range ($-0.79$ at $m=11$, $-0.92$ at $m=512$).

\paragraph{Polylogarithmic constructive-interference advantage.}  The constructive-interference factor $\Dis(m)^2/\Par^\Delta$ grows polylogarithmically.  Fitting independently on the even-$m$ and odd-$m$ subsequences (Corollary~\ref{cor:exactly_one_bipartite}; right panel of Figure~\ref{fig:scaling}) gives
\[
\Dis(m)^2/\Par^\Delta \;\approx\; C\,\log_2 m,
\qquad
C_\mathrm{even} \approx 1.40,\;\; C_\mathrm{odd} \approx 1.39,
\]
with $R^2\approx 0.88$ on the even points and $R^2\approx 0.91$ on the odd points (each fit on $40$ data points).  The two parity classes therefore obey the same logarithmic law with constants agreeing to within $\approx 1.3\%$, and the parity oscillation is fluctuation around a common trend rather than a difference between two scaling laws.  We also tried the alternative one-parameter laws $C\sqrt{\log_2 m}$, $C(\log_2 m)^{3/2}$, and $C\log_2 m\cdot\log_2\log_2 m$; the linear-in-$\log_2 m$ form was the best fit on both parity classes by a clear margin (next-best $R^2$ in the range $0.74$--$0.85$).

The intuition is constructive interference: the difference signal $\Delta f(t)$ is a sum of $\Theta(n^2)$ complex exponentials with amplitudes $\Theta(1/n^2)$; at the optimal time $t^*$, $\Theta(n^2/\log n)$ of the spectral phases align constructively, producing a peak that exceeds the root-mean-square level $\sqrt{\Par^\Delta}\sim 1/n$ by a factor of $\sqrt{\log n}$.  Figure~\ref{fig:diagnostic} shows a concrete example at $m=32$: the typical fluctuation level $\sqrt{\Par^\Delta}\approx 0.027$, the peak value is $\Dis\approx 0.069$, a ratio of~$2.6$ matching the prediction $\sqrt{1.4\log_2 32}\approx 2.65$.  The logarithmic growth reflects the increasing difficulty of global phase alignment as the number of components grows.

Combining the empirical scalings $\Par^\Delta = (3-o(1))/n^2$ and $\Dis(m)^2/\Par^\Delta = \Theta(\log_2 m)$ with the measurement budget of Theorem~\ref{thm:budget} gives the sample size
\begin{equation}\label{eq:budget_scaling}
N_{\mathrm{rep}} \;=\; \frac{2z^2}{\Dis(m)^2} \;\approx\; \frac{z^2 n^2}{2\log_2 n} \;=\; \widetilde O\!\left(\frac{n^2}{\log n}\right).
\end{equation}
Each shot evolves the Hamiltonian for time $T_{\max}=m^2=\Theta(n^2)$ on a $(D{+}1)$-sparse oracle, at quantum cost $\widetilde O(n^2)$ per shot by sparse Hamiltonian simulation~\cite{berry15hamiltonian}.  The total quantum cost is therefore $N_\mathrm{rep}\cdot\widetilde O(n^2) = \widetilde O(n^4)$.  We record this as our central conjecture of efficiency.

\begin{conjecture}[Quantum efficiency for prism vs.\ M\"obius]\label{conj:efficiency}
For the prism-vs-M\"obius identification problem, the cross-graph spectral test of Section~\ref{sec:cross_graph} achieves success probability $\ge 1-\delta$ in $\widetilde O(n^2/\log n)$ oracle queries and $\widetilde O(n^4)$ quantum gates.
\end{conjecture}

The conjecture is supported by $80$ data points spanning $m=4$ to $m=5121$ (so $n$ up to~$10242$), with cross-validation between the direct method (theory-free ground truth) and the SERF method on the overlap range $m\le 33$.  Among the candidate one-parameter laws of the form $C(\log m)^p$, $0<p\le 2$, the data favour $p=1$ but do not formally exclude $\sqrt{\log m}$, $(\log m)^{3/2}$, or $\log m\cdot\log\log m$ corrections; the conjectured $\widetilde O(n^4)$ runtime allows for any such polylogarithmic factor.

\begin{figure}[ht]
    \centering
    \includegraphics[width=1.0\textwidth]{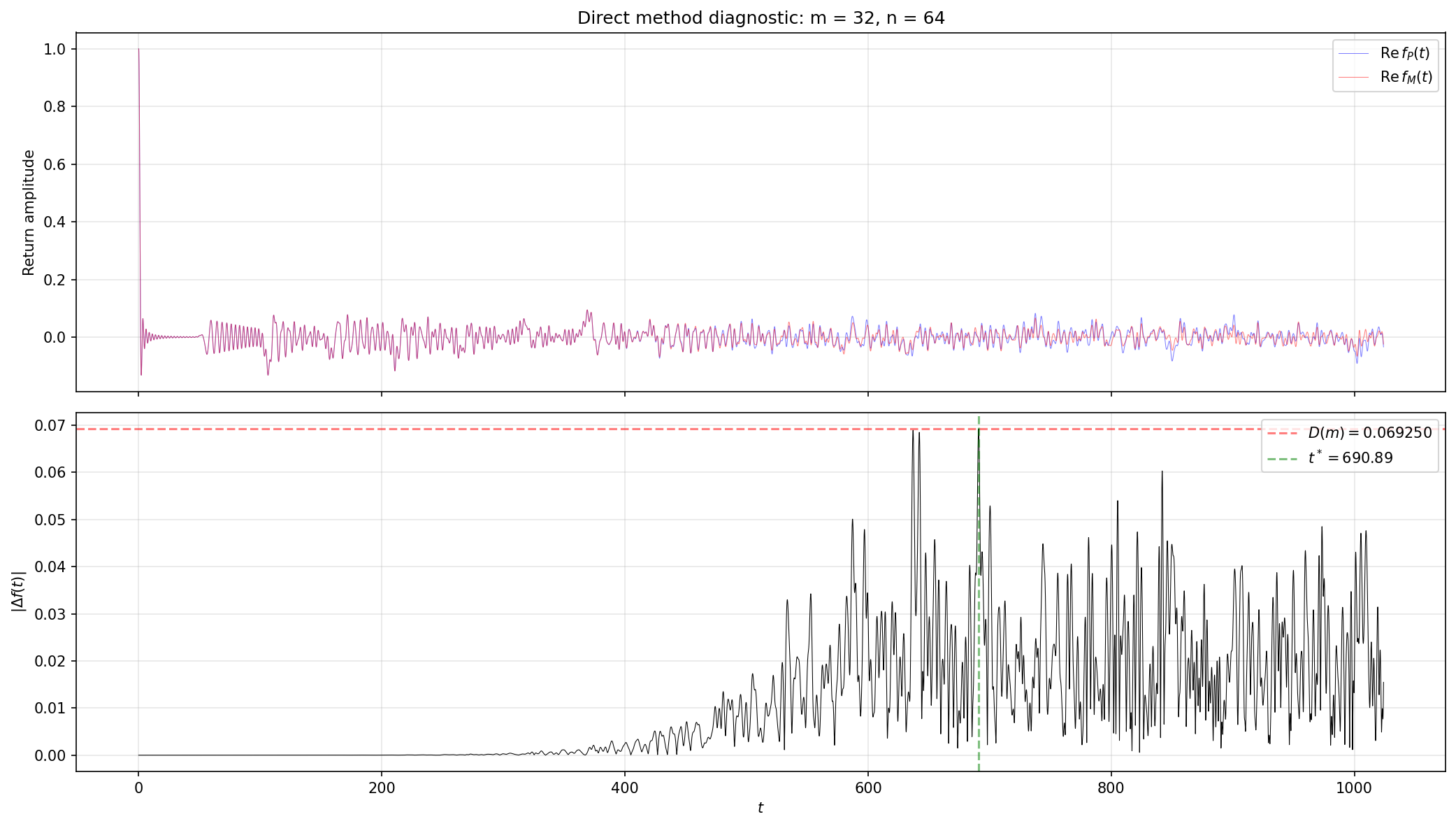}
    \caption{Diagnostic plot for $m=32$ ($n=64$) with $T_{\max}=m^2=1024$, generated by the direct method. \emph{Top:} real parts of the return amplitudes $\mathrm{Re}\,f_Y(t)$ and $\mathrm{Re}\,f_M(t)$. \emph{Bottom:} the distinguishability $|\Delta f(t)|=|f_Y(t)-f_M(t)|$, peaking at $t^*\approx 690.89$ with $\Dis(m)\approx 0.06925$, a factor of $\sqrt{1.4\log_2 32}\approx 2.6$ above the typical fluctuation level $\sqrt{\Par^\Delta}\approx 0.0270$, illustrating the constructive-interference picture.}
    \label{fig:diagnostic}
\end{figure}

\FloatBarrier


\section{Conjecture of classical hardness}\label{sec:hardness}

The quantum algorithm of Sections~\ref{sec:algorithm}--\ref{sec:numerics} identifies a hidden $d$-regular base graph using $O(n^2/\log n)$ queries to the obfuscating oracle. We now turn to the classical side. The construction of $G_\spire$ in Section~\ref{sec:construction} was designed precisely so that the base graph $G$ is hidden from any classical algorithm with only oracle access: spires of height~$L$ separate the apex of the distinguished spire from the lifted-graph layer where the structure of $G$ lives, and the random alternating cycles $C_e$ together with the random labelling $\iota$ destroy all positional information. We conjecture that this hiding is exponentially effective.

This section is organized around the cleanest concrete instance of the question (distinguishing $\cO_{Y_m}$ from $\cO_{M_m}$ for $m$ odd) rather than around the most general statement. Specialising to a single concrete pair sharpens the conjecture: the only difference between $Y_m$ and $M_m$ is the parallel-vs-twisted choice of two closing rail edges, hidden inside an exponentially large spired graph, accessed only through the obfuscating oracle.

\subsection{The distinguishing problem}\label{sec:distinguishing}

Fix two $d$-regular base graphs $G_A, G_B$ on the same vertex set $V$ with the same distinguished vertex $u\in V$, and the security parameter $L\ge 1$ of Section~\ref{sec:oracle}.  Nature samples a bit $b\in\{A, B\}$ uniformly at random, builds the spired graph for $G_b$ with fresh independent random cycles $\mathbf{C}$ and labelling $\iota$, and exposes the resulting obfuscating oracle $\cO_{G_b}$.  A classical algorithm $\cA$, given the seed label $\iota(a_{u})$, makes at most $t$ adaptive oracle queries and outputs a guess $\hat b\in\{A, B\}$.  Its \emph{advantage} is
\[
\Adv(\cA) \;:=\; \Pr[\hat b = b] - \tfrac{1}{2}.
\]

We are interested in how $\Adv(\cA)$ scales with the query budget $t$, the security parameter~$L$, and the graph sizes.

\subsection{The conjecture}\label{sec:conjecture}

\begin{conjecture}[Classical exponential lower bound]\label{conj:hardness}
Let $m$ be odd, $m\ge 5$, and set
\[
G_A \;:=\; Y_m,\qquad G_B \;:=\; M_m,\qquad u \;:=\; \bigl((m{-}1)/2,\, 0\bigr),
\]
in the pair coordinates of Section~\ref{sec:pair_coords}, with security parameter $L\ge 1$ arbitrary.  There exists an absolute constant $C>0$ such that every classical $t$-query algorithm $\cA$ with access to the obfuscating oracle $\cO_{G_b}$ satisfies
\[
\Adv(\cA) \;\le\; C\cdot \frac{t^2}{D^L}.
\]
In particular, achieving constant advantage requires
\[
t \;=\; \Omega\bigl(D^{L/2}\bigr) \;=\; \Omega\bigl(6^{L/2}\bigr),
\]
exponential in the security parameter~$L$.
\end{conjecture}

\begin{remark}[Why $m$ odd and $u$ symmetric]
By Proposition~\ref{prop:placement}, with $m$ odd the distinguished vertex $u=((m{-}1)/2, 0)$ sits at equal graph distance $(m{-}1)/2$ from the two outer-rail closing-edge endpoints $(0, 0)$ and $(m{-}1, 0)$, and at this maximal distance the two ``sides'' of the graph (clockwise and counter-clockwise from~$u$) are interchangeable.  This is the worst case for any classical adversary: the parallel-vs-twisted distinction lives entirely in the differing $4$-cycle, which is sandwiched between two indistinguishable halves of the rail.  For $m$ even this symmetric placement is unavailable (Proposition~\ref{prop:placement}, last clause): the closest closing-edge endpoint sits one step closer than the farther one, breaking the left-right symmetry.  Conjecture~\ref{conj:hardness} is stated for the $m$-odd case where the symmetric placement is genuinely available; the quantum algorithm of Section~\ref{sec:algorithm}, by contrast, is insensitive to the parity of~$m$ and applies at every $m\ge 4$.  We expect the same exponential classical lower bound to hold for $m$~even, with a proof that controls the $O(1)$ rail asymmetry of Proposition~\ref{prop:placement}.
\end{remark}

\begin{remark}[Comparison to the welded-trees problem]
The welded-trees problem of Childs et~al.~\cite{childs03exponential} is structurally close to our setting at $G=K_2$: it shares the same \emph{obfuscation mechanism} (random alternating cycles plus random labelling) and is classically hard via an exponential traversal lower bound.  It is not literally a special case of our identification problem (at $G=K_2$ there is only a single base graph, so nothing to identify); Conjecture~\ref{conj:hardness} extrapolates the welded-trees hardness from traversal to identification.
\end{remark}

\begin{remark}[Generalisation beyond $Y_m$ vs $M_m$]
We expect that exponential quantum-classical separations of the type stated by Conjecture~\ref{conj:hardness} extend to many other pairs of $d$-regular base graphs, plausibly including pairs differing only in a small set of local edges with the distinguished vertex placed far from them.  Identifying the precise structural conditions on a graph pair that enable both the quantum efficiency and the conjectured classical hardness is an open question for future work.  Conjecture~\ref{conj:hardness} is stated only for the prism-vs-M\"obius pair because here the shared common-edge structure (Section~\ref{sec:Y_M}) makes the geometric setup and the symmetric placement of $u$ unambiguous.
\end{remark}


\section{Conclusions and open problems}\label{sec:conclusions}

Conjectures~\ref{conj:efficiency} and~\ref{conj:hardness} together would yield an exponential quantum--classical separation for hidden-base-graph identification in the black-box model.  We close with several open directions stemming from this work.

\paragraph{Other graph families.} Our numerics focus on prism vs.\ M\"obius ladder graphs ($3$-regular, $n=2m$).  Both are vertex-transitive, which simplifies the analysis by ensuring that the return amplitude $f_G(t)$ is the same for every starting vertex.  More generally, the spectral decomposition and the SERF method apply without modification to any pair of $d$-regular base graphs (only the adjacency matrix~$A$ changes), so one could study arbitrary vertex-transitive graphs.  Natural extensions include higher-degree regular graphs ($d\ge 4$), where $D=2d$ changes the path weight $\gamma=\sqrt{D}/2$, and graphs with richer spectral structure.  Ramanujan graphs, which are optimal expanders in the sense that all non-trivial eigenvalues satisfy $|\mu|\le 2\sqrt{d-1}$, are a particularly interesting family: their tight spectral concentration may produce different scaling behavior for $\Dis(m)$ and~$\Par^\Delta$ compared to prisms and ladders, which have eigenvalues spread across the full interval $[-d,d]$.

\paragraph{Alternative expansion structures.} In our construction, the spires (balanced $D$-ary trees of height~$L$) play a dual role.  On the quantum side, a walker starting at the single apex~$a_v$ coherently spreads its amplitude down the spire to all $D^L$ copies of~$v$ at the foundation, so the effective dynamics reduce to the polynomial-dimensional towered graph~$G_\tower$.  On the classical side, the spires confuse any local algorithm: a classical walker cannot distinguish between retreating toward the apex, moving through the lifted edges of the base graph, or entering a spire attached to a neighboring vertex.  The natural question is whether structures other than balanced spires can serve this dual purpose.  Any replacement must (i)~funnel the quantum amplitude from a single entrance vertex into an equal superposition over all copies of a base-graph vertex, producing a low-dimensional invariant subspace, while (ii)~remaining locally indistinguishable from the rest of the graph to a classical observer.  Understanding which alternatives preserve (or break) the quantum speedup is an important structural question.

\paragraph{Identification vs.\ traversal and the exit-finding problem.} Our work differs from the welded-trees problem~\cite{childs03exponential} and its generalizations~\cite{balasubramanian23hierarchical,li23pathfinding} in a fundamental way: those results solve \emph{traversal} problems (reaching a designated exit vertex), whereas our algorithm solves an \emph{identification} problem (determining which base graph is hidden).  A natural problem that bridges the two is \emph{exit-finding}: given oracle access to~$G_\spire$ and the seed apex label, find the remaining apices (the apices corresponding to base-graph vertices other than the seed).  For different base graphs, the number and arrangement of apices differ, and the difficulty of exit-finding may depend on the spectral properties of the hidden graph.

\paragraph{Beyond black-box separations.} Like the welded-trees result of Childs et~al.~\cite{childs03exponential} and the subsequent hierarchical-graph and pathfinding extensions~\cite{balasubramanian23hierarchical,li23pathfinding}, our work is a \emph{black-box separation} (Conjectures~\ref{conj:efficiency} and~\ref{conj:hardness} are relative to the obfuscating oracle).  The open challenge in quantum algorithms is to translate such oracle separations into concrete computational advantages, as the hidden subgroup framework does for the discrete logarithm problem (where Shor's algorithm achieves an exponential speedup for a problem of direct cryptographic relevance).  It remains an important goal to find a natural computational problem whose mathematical structure contains an exponentially large graph with hidden global symmetry of the kind described above.  Detecting that symmetry would translate the oracle separation into a concrete computational advantage.


\section*{Acknowledgements}

The author would like to thank Clive Elphick for many enjoyable joint projects on spectral graph theory over the years, which shaped the perspective behind this work. Thanks also to Andrew M.\ Childs, Matthew Coudron, and Gabriel Coutinho for helpful discussions.


\end{document}